\documentclass[11pt,oneside,onecolumn,draftcls]{IEEEtran}
\usepackage[dvips,final]{graphicx}
\usepackage{latexsym}
\usepackage{amssymb}
\usepackage{amsfonts}
\usepackage{amsmath,epsfig,algorithmic,color}
\usepackage{algorithm}
\usepackage{cite}
\usepackage{textcomp} 
\usepackage{epstopdf} 
\usepackage{bm}  
\usepackage{epsfig}
\usepackage{wrapfig}
\usepackage{multirow}
\usepackage[caption=false,font=footnotesize]{subfig}
\usepackage{setspace}
\usepackage{pdfpages}

\newtheorem{proposition}{{Proposition}}

\DeclareMathOperator{\rank}{rank}
\newcommand{\argmin}{\operatornamewithlimits{argmin}}

\newcommand{\qbinom}[2]{{\genfrac{[}{]}{0pt}{}{#1}{#2}}_{q}}
\title{Adaptive Prioritized Random Linear Coding and Scheduling for Layered Data Delivery from Multiple Servers}
\author{Nikolaos Thomos, \textit{Member, IEEE}, Eymen Kurdoglu, \textit{Student Member, IEEE}, Pascal Frossard, \textit{Senior Member, IEEE}, and Mihaela van der Schaar, \textit{Fellow, IEEE}
\thanks{N. Thomos is with the Computer Science and Electronic Engineering Department, University of Essex, Colchester, United Kingdom (e-mail: nthomos@essex.ac.uk).}
\thanks{E. Kurdoglu is with the Polytechnic Institute of NYU, NY, USA, USA (e-mail: eymen@vision.poly.edu.)}
\thanks{P. Frossard is with the Signal Processing Laboratory 4 (LTS4), Ecole Polytechnique F\'ed\'erale de Lausanne (EPFL), Lausanne, Switzerland (e-mail: pascal.frossard@epfl.ch).}
\thanks{M. van der Schaar is with the Networks, Economics, Communication Systems, Informatics and Multimedia Research Lab, UCLA, CA, USA (e-mail: mihaela@ee.ucla.edu).}
\thanks{This work has been initiated while the first two authors were at EPFL. It has been supported by the Swiss National Science Foundation, under grants PZ00P2-121906, PZ00P2-137275, and the project ``Adaptive Network Coding for Video Communications'' funded by Hasler foundation.}
}

\begin{document}

\maketitle

\begin{abstract}
In this paper, we deal with the problem of jointly determining the optimal coding strategy and the scheduling decisions when receivers obtain layered data from multiple servers. The layered data is encoded by means of Prioritized Random Linear Coding (PRLC) in order to be resilient to channel loss while respecting the unequal levels of importance in the data, and data blocks are transmitted simultaneously in order to reduce decoding delays and improve the delivery performance. We formulate the optimal coding and scheduling decisions  problem in our novel framework with the help of Markov Decision Processes (MDP), which are effective tools for modeling adapting streaming systems. Reinforcement learning approaches are then proposed to derive reduced computational complexity solutions to the adaptive coding and scheduling problems. The novel reinforcement learning approaches and the MDP solution are examined in an illustrative example for scalable video transmission. Our methods offer large performance gains over competing methods that deliver the data blocks sequentially. The experimental evaluation also shows that our novel algorithms offer continuous playback and guarantee small quality variations which is not the case for baseline solutions. Finally, our work highlights the advantages of reinforcement learning algorithms to forecast the temporal evolution of data demands and to decide the optimal coding and scheduling decisions.
\end{abstract}


\maketitle

\section{Introduction} 

The emergence of new multimedia devices such as laptops or smartphones has motivated the research on efficient data delivery mechanisms that are able to adapt to dynamic network conditions and heterogeneous devices' capabilities. The quality of experience of the different receivers depends on their display size, processing power, network bandwidth, \textit{etc.} In order to accommodate for such a diversity, the data is progressively encoded in several quality layers. This permits to offer a basic quality to receivers with limited capabilities, while other devices can have higher quality of experience. Further, the layered encoding of the data helps data delivery systems to cope with the strict delivery deadlines as data quality can be adapted to meet decoding deadlines when network resources get scarce.

Methods such as Fountain codes (Random Linear Coding \cite{RandomizedNC03}, Raptor codes \cite{Shokrollahi06}, LT codes \cite{Luby02}) have been proposed recently for efficient data delivery in error-prone networks. Usually, Fountain codes are applied at servers, and the receivers obtain the requested data from possibly multiple servers. The use of such codes leads to significant performance gains \cite{WagnerICME06,SchierlJVCIR08} compared to classical channel codes because channel protection does not have to be determined beforehand. If the set of clients is very heterogeneous, the source data can be encoded in a scalable way for adaptive delivery. This can be achieved with the use of embedded (prioritized) codes such as Expanding Ä Fountain (EWF) Codes \cite{EWF}, Prioritized Random Linear Codes (PRLC) \cite{PriorityRNC} and Prioritized Random Linear Network Codes (PRLNC) \cite{ThomosTMM2011}. These schemes encode the layered data progressively and form classes of packets with unequal importance. However, in order to optimize the decoding quality an optimized coding algorithm along with a proper scheduling mechanism and inter-server coordination are needed for efficient resources allocation. 

In this paper, we focus on the multi-server time-constrained transmission of layered data encoded with PRLC. PRLC is employed to respect the unequal importance of the data layers. We consider a novel framework where the data is encoded in groups of elements (\textit{i.e.}, generations) that have the same decoding deadline, which are transmitted concurrently to the receiver in order to improve performance and decrease the decoding delay. Our framework is different from common rateless and network coding based systems \cite{BoginoISCAS07,CataldiTIP10,AhmadTMM11,SeferogluJSAC2009,Nguyen11}, where data is sequentially transmitted. In this context, we devise a receiver driven protocol in order to communicate information about the link statistics and the buffers states of the different servers and decide on the data requests. Then, we investigate a new problem of deciding the optimal coding and scheduling policies for layered data transmission such that the quality at the receivers is maximized, while the decoding deadlines are respected. We first propose a data-aware optimization framework based on Markov Decision Processes (MDPs) \cite{BertsekasDynPro} in order to characterize the performance of our adaptive streaming system. While MDPs offer an elegant way to describe our novel framework, they require high computational complexity and perfect knowledge of the environment in order to describe the state transition probabilities and reward matrices. In order to cope with these limitations, we propose a variant of the Q-learning approach \cite{QLearning} that is inspired from the ``virtual experience'' concept introduced in \cite{MastronardeTCOMP2013}. We start with the observation that channel dynamics are independent from the states and actions, which in our formulation represent the buffer content of the receiver and the requested packets, respectively. Then, we exploit the knowledge acquired by observing a state-action pair by extending it to multiple statistically equivalent states-action pairs. This new algorithm trades off convergence rate of the optimization algorithm with computational complexity. Like Q-learning, our solution does not require precomputing and storing the transition probabilities and reward matrices. It however converges faster than Q-learning, with a higher computational complexity. We present simulation results for an illustrative scalable video transmission example. These results show the advantages of the proposed method compared to myopic methods for coding and scheduling and to the original Q-learning, in terms of video quality and quality fluctuations. Moreover, our scheme offers small quality variations as it efficiently considers the data importance in the coding and scheduling decisions. This is an important advantage in real-time video communication systems. Further, the proposed Q-learning variant improves the data quality with only minimal additional computational cost compared to Q-learning. Overall, the proposed solution permits easy adaptation to changing network conditions, as the optimal coding and scheduling decisions can be continuously updated based on the actual network dynamics.

The paper is organized as follows. In Section~\ref{sec:relatedwork}, we provide a brief discussion of Fountain and network coding communication systems that target the delay-constrained delivery of data. Next, we present in Section \ref{sec:prnc} the basic principles of the Prioritized Random Linear Coding and describe our receiver-driven communication protocol. In Section \ref{sec:MDP}, we formulate our joint PRLC coding and scheduling optimization problem with the help of Markov Decision Processes. Then, we describe our reduced complexity solution that is based on Q-learning in Section \ref{sec:overallqlearn}. In Section \ref{sec:expres}, we analyze our framework for an illustrative video application with multiple servers and a single receiver and evaluate its performance for transmission of layered video data. Finally, conclusions are drawn in Section \ref{sec:conclusions}.

\section{Related work}
\label{sec:relatedwork}

In this section, we briefly overview and position our work with respect to other systems that employ Fountain codes and RLC for delivery of time-constrained data from single and multiple sources. 
  
Streaming of scalable video from multiple servers has been investigated in \cite{WagnerICME06}, where data from each Group of Pictures (GOP) and layer is independently protected by means of Raptor codes. A rate allocation algorithm determines the optimal number of Raptor packets to be sent from each GOP and layer. Similar concepts have been presented in \cite{SchierlJVCIR08} for transmission of scalable video in mobile ad hoc networks (MANETS). Video on Demand delivery in wireless mesh networks has been studied in \cite{DingTON12} where the video streams are protected with Reed Solomon codes prior to transmission. Then, the streams are transmitted over multiple disjoint paths that are discovered by two heuristics algorithms. The optimal rate per path is decided by a data unaware optimization algorithm. The problem of time-constrained data streaming from a single server has been more extensively studied in the literature. Although the extension from a single server to multiple servers is not trivial, the investigation of single server systems sheds light on the properties that efficient multi-source streaming schemes should have. 

A variety of Application Layer Forward Error Correction (AL-FEC) schemes appropriate for transmission of scalable video over broadcast erasure channels have been recently presented in \cite{BoginoISCAS07,CataldiTIP10,AhmadTMM11,VukobratovicTMM09,WagnerICME06,SchierlJVCIR08}. In \cite{BoginoISCAS07}, the sliding window concept has been introduced and it has been later combined in \cite{CataldiTIP10} with Unequal Error Protection (UEP) Raptor codes for providing enhanced error robustness. This scheme has been further improved in \cite{AhmadTMM11}, where data replication is used prior to the application of Fountain codes. This results in stronger protection for the most important layers. The expanding window approach has been proposed in \cite{VukobratovicTMM09} for video multicast as an alternative to the sliding window method.

Unfortunately, digital Fountain codes such as Raptor codes and LT codes may perform poorly in terms of delay \cite{DrineaPhyCom13,NistorJSAC11}. To this aim, systematic Random Linear Codes (RLC) with feedback have been studied in \cite{DrineaPhyCom13} for minimizing the average delay in database replication systems. The feedback messages contain information about the packets that have been delivered to the receivers. These messages are used to optimize the coding decisions such that the decoding delay is small. The decoding delay distribution has been investigated in \cite{NistorJSAC11} for RLC systems. Markov chains are utilized in order to find the optimal coding solution for the case of two receivers. However, this problem is computationally intractable for three or more receivers and hence only a heuristic solution is presented. The delay benefits of a RLC based scheme have been examined in \cite{TournouxTMM11}, where repair packets are generated on-the-fly through RLC according to feedback messages that are periodically received by the servers. This scheme copes efficiently with the packet erasures and is appropriate for real-time sources such as video. From the above studies, the advantages of RLC over Fountain codes in terms of delay become clear. 

Although the employment of RLC may result in large delay gains, scheduling of RLC encoded data is a NP-hard problem \cite{DrineaPhyCom13}. Since it is hard to find the optimal coding and scheduling solution, a number of practical solutions in wireless networks that exploit the feedback messages that are available to each network node are proposed in \cite{SeferogluJSAC2009,YangMobiHoc12,Nguyen11}. In \cite{YangMobiHoc12}, the size of the coding generations is adapted by considering hard deadlines imposed by the target applications. It is shown that the generation size should become smaller as the delivery deadline approaches. The joint coding and scheduling problem is studied in \cite{Nguyen11} for video transmission in wireless networks, where coding operations are performed in binary fields. The problem is modeled as MDP and is solved by a dynamic programming algorithm. The employed XOR-based coding permits instantaneous decoding but with some performance penalty compared to schemes that work in larger finite fields. The approaches in \cite{YangMobiHoc12,Nguyen11} provide some interesting directions towards the development of adaptive data delivery mechanisms. However, their applications are rather limited as the work in \cite{Nguyen11} is appropriate only for small generations and the work in \cite{YangMobiHoc12} is data oblivious. Both works are myopic and process sequentially the generations which generally penalizes the performance compared to foresighted strategies. Finally, the work in \cite{HalloushTWC} proposes an interesting framework with concurrent generation delivery. Two major problems are associated with this approach: (a) the decoding computational cost and (b) the coding overhead \cite{PracticalNC03} (a header with the coding coefficients that should be communicated along with the payload of the packets). The decoding computational cost and the overhead increase with the number of generations involved in the encoding operations. This prohibits the use of the approach in \cite{HalloushTWC} in real-time applications. The system that we propose in this paper does not have such limitations.

Our work departs from the common choice of sequential processing of generations and considers simultaneously coding and scheduling of data from multiple generations. It aims to find the optimal PRLC coding and scheduling decisions so that the data quality is maximized and the playback of streaming data is smooth at receivers. To the best of our knowledge, this strategy has not been investigated for PRLC. Our system is inspired by the recent advances in cross layer optimization for wireless communication systems \cite{MastronardeTCOMP2013,FuJSAC2010}. In \cite{FuJSAC2010}, the joint scheduling and rate allocation problem of wireless video from multiple servers is modeled as Multi-User MDP (MUMDP) and then is effectively solved with online reinforcement learning methods \cite{ReinforcementLearning} where multiple states are simultaneously updated using stochastic sub-gradient methods \cite{BertsekasNonLinear}. The work in \cite{MastronardeTCOMP2013} makes joint physical layer and system level power management decisions for delay sensitive applications. It uses the concepts of the ``postdecision state'' \cite{SalodkarJSAC08} in order to decompose the problem in two learning problems: one that factors the known dynamics and another that has the unknown dynamics. Further, it proposes the concept of ``virtual experience'' in order to accelerate the learning process.

\section{Prioritized delivery of PRLC encoded data}
\label{sec:prnc}

\subsection{Network model}

We consider a framework where receivers request layered data from multiple servers. The time is slotted and receivers send their data requests to all their parents simultaneously. The servers then encode the data in order to meet the requests of receivers, and deliver coded packets of uniform size to the receivers. The link connecting a server $i$ with the receiver $j$ is denoted by the tuple $(i,j)$ and is characterized by three parameters: the transmission rate $f_{ij}$ expressed in packets per second, the packet loss probability $\epsilon_{ij}$ and the transmission delay $\eta_{ij}$. We assume that $f_{ij}$ and $\epsilon_{ij}$ do not change very frequently, at least not faster than the optimization decisions. This is a realistic assumption in practice, which permits the convergence of our algorithm. 

Due to nodes' limited upload and download capabilities, all the servers are not necessarily available for every receiver. Hence, the set of servers for the receiver $j$ is represented by
\begin{align}
P_{j}=\{i: (i,j)\in E\},\nonumber
\end{align}
where $E$ is the set of links, which are in general non-homogeneous in terms of their loss rates. The link parameters, as well as the amount of data available at servers, are periodically sent to receivers. Finally, the transmission delays on each link are non-homogeneous and modeled as \textit{i.i.d.} random variables.

\subsection{Layered data}

Since servers and receivers are heterogeneous devices with diverse capabilities, the data is progressively encoded in $L$ quality layers, \textit{i.e.}, a base layer and several enhancement layers. Such an encoding enables receivers with limited capabilities to receive the base layer, while receivers with more resources decode the data in higher quality. The data from the base layer is the most important, followed by the data of the successive enhancement layers that offer incrementally finer levels of quality. A receiver can decode data from an enhancement layer only when it has decoded all the lower layers. 

The data is further segmented in generations $G_0, G_1,\ldots, G_k$, which are groups of time-constrained data (\textit{e.g.}, images in video case). We consider that the data of the $l$th layer of a generation is packetized into $\alpha_{l}$ packets. Therefore, the first $k$ layers consist of $\beta_{k}=\sum\limits_{l=1}^{k}{\alpha_{l}}$ packets. The distortion reduction associated with the decoding of the data layer $l$ of a generation $G_m$ is denoted as $\delta_{l,m}$, while the cumulative distortion reduction after decoding  the first $l$ layers is $\Delta_{l,m}=\sum\limits_{i=1}^{l} \delta_{l,m}$.

Each generation has in total $\beta_{L}$ packets and is associated with a decoding delay deadline. The deadline of the $n$th generation is written as
\begin{equation}
D_{n}=D_{0}+n D_{G},
\label{eq:Delayn}
\end{equation}
where $D_G$ is the duration of the rendering of the information contained in a generation. We focus on the case where $D_G$ is identical for all generations. The parameter $D_{0}$ in (\ref{eq:Delayn}) represents the playback delay and depends on the application's requirements.

\subsection{Prioritized Random Linear Coding}
\label{sec:PNC}
 
We consider that the servers encode the data with PRLC codes \cite{PriorityRNC} in order to respect the unequal levels of importance in the data layers. The encoding of the data permits to deal with losses in the network; at the same time, it avoids the need for precise coordination among the multiple servers transmitting data to the same receiver. All the PRLC operations are performed in a finite field $\mathbb{F}_{q}$ of size $q$. We restrict the coding operations to packets that belong to the same generation. Then, for a generation of $\beta_{L}$ packets, an overhead of length $\beta_{L}$ symbols is appended to each packet that belongs to that particular generation, in order to convey the coding coefficients to the receiver. Similarly to the Prioritized Randomized Network Coding method presented in \cite{ThomosTMM2011,KurdogluICME11} we consider that a packet of class $l$ is a random linear combination of data from the first $l$ video layers. Finally, the type of each packet is described by a tuple $(G_{m},l)$, where $G_m$ and $l$ stand for the generation $m$ and the class $l$, respectively. 

Upon the arrival of the PRLC packets at the receiver, the packets are sorted according to their generation $G_m$. These packets are examined in terms of their innovation (\textit{i.e.}, whether they bring innovative information with respect to the information that is already available at the receiver). If they are innovative, they are stored in the decoding decoding matrix $\bm{W}_{m}$, otherwise they are dropped. Each row of the decoding matrix $\bm{W}_{m}$ corresponds to an innovative packet  from generation $G_m$. This row contains the coding coefficients contained in the header of the received packet. Although the packets do not necessarily arrive in order, we assume that the rows of the $\bm{W}_{m}$ matrix are sorted in an ascending class order. Therefore, packets from the first class are placed on the top of the matrix, followed by packets from the second class, \textit{etc}. The matrix $\bm{W}_{m}$ is divided into $L$ submatrices $\bm{F}_{lm}$, where $\bm{F}_{lm}$ is a $n_{lm} \times \beta_{L}$ matrix and $n_{lm}$ denotes the number of received packets of class $l$ from generation $G_m$. Recall that $\beta_l$ is the number of packets from the first $l$ layers of a generation. This is illustrated in Fig. \ref{fig:decmat}.

\begin{figure}[t]
\centering
\includegraphics[width=0.6\textwidth]{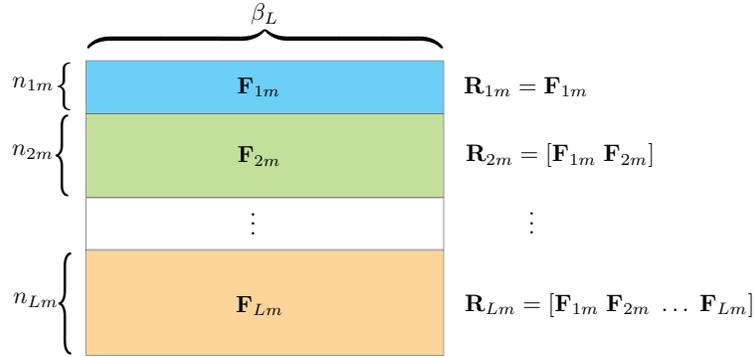}
\caption{Decoding matrix $\bm{W}_{m}$ of generation $G_m$ constructed at the receivers. It consists of submatrices $\bm{F}_{im}$ that contain the innovative data from the $i$th video layer of the generation $G_m$. Matrices $\bm{R}_{lm}$ are the unions of the first $l$ submatrices $\bm{F}_{1m},\dots,\bm{F}_{lm}$, whose number of rows $n_{im}$ denotes the number of innovative packets from the layer $i$ of the generation $G_m$. Decoding of \textit{only} the first  $l$ layers of the generation $G_m$ is possible when $\sum\limits_{i=1}^{l} n_{im}=\beta_l$ and $\sum\limits_{j=1}^{i} n_{jm} \le \beta_i, \forall i>l$.}
\label{fig:decmat}
\end{figure}

Within the matrix $\bm{W}_{m}$ one can form submatrices $\bm{R}_{lm}$ that are the unions of the first $l$ submatrices $\bm{F}_{1m},\dots,\bm{F}_{lm}$. For these submatrices, it holds that $0 \leq r_{lm} \leq \beta_{l}$, where $r_{lm}=\rank(\bm{R}_{lm})$, as the number of innovative packets from class $l$ in generation $G_m$ cannot exceed the cumulative number of source packets from the first $l$ layers. When $r_{lm}=\beta_{l}$, a receiver can successfully decode the first $l$ layers of the generation $G_m$ by means of Gaussian elimination, for example. 

The rank vector $\bm{r}_{m}=(r_{1m},\dots,r_{Lm})$ for the submatrices $\bm{R}_{lm}$ of the generation $G_{m}$, completely defines the buffer state of the receiver. For notational convenience we define hereafter the buffer state of a receiver $i$ as $B_{i}=\{\bm{r}_{m}\}$, \textit{i.e.}, $B_{i}$ is the collection of the $\bm{r}_{m}$ vectors (one per generation $G_m$) in the buffer of the $i$th receiver. A similar buffer state variable is used to describe the data available at servers.

\subsection{Receiver-driven delivery optimization}
\label{sec:probform}

In our receiver-driven system multiple servers simultaneously transmit data from different generations to the receiver. When the receiver estimates that it cannot receive all the data of a particular generation at a given rate by its decoding deadline, it either decides to request lower data rate, hence lower data quality or to skip the particular generation and proceed with the next one. With their decisions, the receivers naturally attempt to obtain high data quality with small quality fluctuations. 

The receiver periodically computes the optimal coding and scheduling strategies taking into account the local statistics, \textit{i.e.}, the channel loss rate $\epsilon_{ij}$ and the capacity $f_{ij}$ of each link connecting the receiver to the available servers. It also considers information about the available data at each server (\textit{i.e.}, the buffer maps $B_{i}[n]$ of the server $i$ at time $t=t_n$). This information becomes available to the receivers with the form of messages $\mathcal{I}_i={(B_{i}[n], \epsilon_{ij}, f_{ij},\eta_{ij}})$ that are transmitted from the servers $i \in P_j$ to the receiver $j$. The receiver $j$ takes into account these messages $\mathcal{I}_i$ to calculate a request allocation vector $v_j=[v_{ij}, \;\forall i \in P_j]$, which is sent to all the servers. The request allocation vector for the $i$th server is written as $v_{ij}=[v_{11}^{ij},\ldots,v_{1L}^{ij},\dots, v_{l1}^{ij},\ldots,v_{lL}^{ij},\dots, v_{m1}^{ij},\ldots,v_{mL}^{ij}]$, where $v_{lk}^{ij}$ denotes the number of packets of class $k$ in generation $G_l$ that receiver $j$ requests from server $i$. The request allocation vector $v_j$ is computed at the receiver in order to maximize the cumulative expected distortion reduction in the $n$th time slot for all the generations that have not expired. The overall data delivery system is illustrated in Fig. \ref{fig:comscen}.

\begin{figure}[t]
\centering
\includegraphics[width=0.5\textwidth]{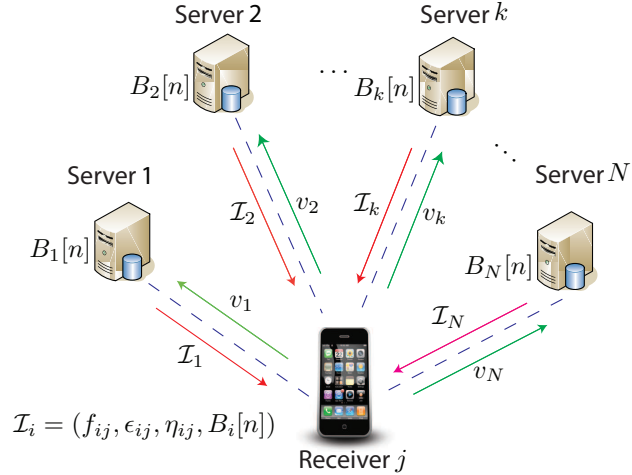}
\caption{Data delivery system under consideration. Multiple servers that have layered data simultaneously transmit packets to a single receiver (node $j$). Each link $(i,j)$ is characterized by the transmission rate $f_{ij}$, the loss rate $\epsilon_{ij}$ and the transmission delay $\eta_{ij}$. The receiver $j$ computes the request allocation vectors $v_i$ based on the $\mathcal{I}_j={(B_{j}[n], \epsilon_{ji}, f_{ji},\eta_{ji}})$}.
\label{fig:comscen}
\end{figure}

\section{MDP formulation}
\label{sec:MDP}

In this section, we propose a Markov Decision Process (MDP) formulation to find the request allocation vector $v_j$ at receiver $j$ that maximizes the cumulative expected distortion reduction given the messages $\mathcal{I}_i$ that are received from the servers $i \in P_j$. The receiver, also called agent in our formulation, periodically determines the optimal request allocation policy $\pi^{*}$ every $T$ time slots. Hence, the first decision is made at $t=t_{0}$, while the $n$th decision is at $t_{n}=t_{0}+nT$. We consider that all receivers then act independently in our system. In the rest of the paper, we drop the index $j$ of the receiver and hence $P=P_j$, $v=v_j$, $\epsilon_{ij}=\epsilon_i$, $f_{ij}=f_i$, $\eta_{ij}=\eta_{i}$ for easy notation. The actions, states and the rewards of the proposed MDP framework are now presented in detail for one receiver. 

\subsection{States}\label{sec:states} 

The state of the receiver at time $t_n$ is completely characterized by its own buffer map $B[n]$ and the buffer maps of the servers $(B_{k}[n],\;\forall k\in P)$. Therefore, the state of the receiver is defined as 

\begin{equation}\label{eq:veryrough}
S[n]=\left\{B_{k}[n] \ (\forall k \in P), \ B[n], \ t_n\right\}
\end{equation}

For infinite length data, the states space is enormous. It becomes smaller when the states space is quantized by considering that time is discrete and that not all the generations are active (\textit{i.e.}, decoding deadlines expire during the transmission). 
 
In order to find which generations should be included in the states definition, at time instant $t$ the generations are relabeled based on their expiration time (urgency). The relabeling is done as $G_{m}\xrightarrow{t} G^{*}_{m-\mu(t_n)}$, where $\mu(t_n)$ is the index of the generation with the closest decoding deadline at time $t_n$,

\begin{equation}
\mu(t_n) = \argmin_{m:\ t_n<D_{m}}D_{m}
\end{equation}

Therefore, the remaining time $\tau(n)$ until the deadline of the most urgent generation $G_{\mu(t_n)}$ is denoted as $\tau(n)=D_{\mu(t_n)}-t_n$. Clearly, $\tau(n)$ satisfies $1\leq\tau(n)\leq {D}_{G}$, with $D_G$ the duration of a generation. 

Hence, the state definition of (\ref{eq:veryrough}) can be re-written as follows:

\begin{equation}\label{eq:rough}
S[n]=\left\{B_{k}[n] \ (\forall k \in P), \ B[n], \ \tau[n]\right\},
\end{equation}

We now analyze briefly the size of the states space $\mathcal{S}$. We note that for any generation $G_m$, the number of valid vectors (states) is finite, constant and dependent on the source data, \textit{i.e.}, the number of packets per layer, $\alpha_{l}$. Therefore, the overall number of states $|\mathcal{S}|$ is given by

\begin{equation}\label{eq:validrankvectors}
|\mathcal{S}|=\sum_{v_{1}=0}^{\beta_{1}}\sum_{u_{2}=0}^{\beta_{2}-u_{1}}\dots
\sum_{u_{L-1}=0}^{\beta_{L-1}-\sum_{k=1}^{L-2}u_{k}}
\left(\beta_{L}-\sum_{k=1}^{L-1}u_{k}\right).
\end{equation}

From (\ref{eq:validrankvectors}), we conclude that the states space increases exponentially with the number of packets per layer. The number of generations that might be available in the buffer of the receiver at time $t_{n}$ is determined by the capacity of its incoming links and the decoding deadlines of the generations. We note that we only consider information from the generations that have not expired nor been decoded in the states definition.

\subsection{Actions}\label{sec:actions} 

An action $A_{k}$ is defined as the request allocation vector ($v_k$) transmitted from the receiver to the server $k$. This vector determines, for the server $k$, the class and the generation of the PRLC encoded packets to be sent. Then, an action $A$ consists of actions $A_{k},\; \forall k \; \in \; P$ and is written as

\begin{equation}
A=[A_{k},\; \forall k \; \in \; P]
\end{equation}

The length of the request vector $v_k$ that corresponds to the action $A_{k}$ is equal to the overall number of layers $L\cdot N_m$,  where $N_m$ stands for the number of generations that have not expired and $L$ the number of layers per generation\footnote{We assume that all the generations have the same number of layers}. It is assumed that these vectors arrive at servers without delay. Based on $v_{k}$, each server $k$ generates packets by PLRC coding and transmits them to the receiver. The $l$th packet in the vector is scheduled for transmission at $t=t_{n}+(l-1)f_{k}^{-1}$, where $f_{k}^{-1}$ is the time interval between the transmission of two consecutive packets on the link from the server $k$ to the receiver.

We now study briefly the size of the action space $\mathcal{A}$. Now, if  $\mathcal{A}_k$ is the action space of the server $k$, it holds $\mathcal{A}=\times_{k\in P}\mathcal{A}_{k}$. Given a state $S[n]$, and that the receiver might request packets that belong to $M_{k}$ different classes from the server $k$, hence, the size of the action space $|\mathcal{A}|$ for a given state $S[n]$ is

\begin{equation}
\left|\mathcal{A}\right|=\prod_{k\in P} \binom{N_{k}+M_{k}-1}{N_{k}}
\label{eq:actionsspace}
\end{equation}
where $N_{k}$ is the number of requested packets from server $k$.

\subsection{Rewards}\label{sec:rewards} 

Since the successful decoding of a layer is associated with a reduction in the distortion of the decoded data, the expected reward $J(\cdot)$ at time slot $n$ is expressed as the resulting cumulative distortion reduction for all active generations after the completion of all actions decided at time $t_{n-1}$. The reward $J(S[n],A[n])$ is a non-negative function of the current state-action pair.

Let $\mathcal{G}[n]=\{G_{m}: t_{n}<D_{m}\leq t_{n+1}\}$ be the set of generations that have deadlines in $(t_{n}$,$t_{n+1}]$. If $\mathcal{G}[n]$ is non-empty, $\tau[n]$ is smaller than $T$ and then the distortion reduction might be positive. When $\mathcal{G}[n]\neq\emptyset$, the expected distortion reduction is $\Delta_{\ell,m}=\sum_{k=1}^{\ell}{\delta_{k,m}}$ (with $\Delta_{0,m}=0$), when no layers have been decoded for the generation $G_m$ in the previous decision interval. Then, the reward function becomes

\begin{equation} \label{eq:objectivefunc}
J(S[n],A[n])=\sum_{G_{m}\in\mathcal{G}[n]}\sum_{\ell=0}^{L}{\mathbb{P}_{\ell,m}(S[n],A[n])\Delta_{\ell,m}}, 
\end{equation}
where $\mathbb{P}_{\ell,m}(S[n],A[n])$ denotes the probability that the receiver recovers \emph{only} the first $\ell$ layers of the generation $G_{m}$, that is, $r_{\ell m}=\beta_{\ell}$ and $r_{lm} < \beta_{l},\forall l > \ell$, given $S[n]$ and $A[n]$. This probability can be calculated by considering the subspaces formed by the received packets from different classes. Let us denote by $V_{\ell,m}$ the subspace spanned by the source packets from the first $\ell$ layers of the generation $G_m$. Furthermore, let $B_{\ell,m}^i[n]$ denote the subspace spanned by the innovative packets from the class $\ell$ of the generation $G_m$ that are transmitted by the node $i$ to the receiver.\footnote{This is a slight abuse of the notation $B_{k,m}^i[n]$ which originally denotes the subspace spanned by the innovative packets from the class $\ell$ of the generation $G_m$ that are stored at server $i$.} To compute $\mathbb{P}_{\ell,m}(S[n],A[n])$, we need to consider the events in which the received packets span $V_{\ell,m}$, but cannot span $V_{r,m}$, $\forall r>\ell$. Thus,

\begin{equation}
\mathbb{P}_{\ell,m}(S[n],A[n])=\text{Pr}\left( \bigcup\limits_{i \in P^{'}}  \bigcup\limits_{k=1}^{\ell} B_{k,m}^i[n] = V_{\ell,m}, \bigcup\limits_{i \in P^{'}}\bigcup\limits_{k=\ell+1}^{r} B_{k,m}^i[n] \subset V_{r,m},r\in\{\ \ell+1,...,L\}\right)
\label{eq:old_eq1}
\end{equation}
where $P^{'}$ is the set comprised of the receiver and its  parental servers. 

In order to calculate the probability $\mathbb{P}_{\ell,m}(S[n],A[n])$, we need to consider the probability that the PRLC coded packets received from the servers are innovative. Non innovative packets (\textit{i.e.}, packets that can be expressed as linear combinations of the existing in the buffer packets) can be received as PRLC selects the coding coefficients uniformly at random. These packets can be discarded as only innovative packets bring benefits to the receiver. Based on those considerations, the probability $\mathbb{P}_{\ell,m}$ can be computed by considering the set of packets that are available to the servers and the receiver (we need to know how many packets we require to decode a layer) and the probability to receive as many packets are required for rendering decodable the packets that belong to class $\ell$ of generation $m$. Details are provided in the Appendix \ref{app} and in particular in Proposition 3.

In order to avoid myopic behaviors and greedy decisions, the MDP framework considers future rewards in the optimization of the rate allocation vector. Hence, we consider the expected rewards for future generations as
\begin{equation} 
J_{k}(S[n],A[n])=\sum_{G_{m}\in\mathcal{G}[n^{'}]}\sum_{\ell=0}^{L}{\mathbb{P}_{\ell,m}(S[n],A[n])\Delta_{\ell,m}}, 
\label{eq:Jk}
\end{equation}
with $n^{'}=n+k$ and $k>0$. 

The rewards in (\ref{eq:Jk}) are discounted by a factor $\gamma\in[0,1]$ which is known as ``discount rate'' and controls the contribution of the future rewards (delayed rewards). Hence, the reward of actions $A[n]$ for generations $G_{n+k}$ is weighted by $\gamma^k$. If $\gamma=0$, the decisions are myopic. In contrast, if $\gamma=1$ the agent becomes farsighted and is interested only in future rewards. This is further discussed in the next section.

\subsection{Optimization Problem}
\label{sec:MDPoptim}

With the above MDP framework, we can now define the problem that is solved for the optimal rate vector selection. We want to maximize the total expected discounted reward (value function), which is the discounted sum of the received rewards starting from state $S[0]$ and following a decision policy $\pi$ that defines the action $A[n]$ at state $S[n]$, \textit{i.e.}, $A[n]=\pi(S[n])$. Formally, the value function is given by

\begin{equation}
V^{\pi}(S[n])=\sum_{k=0}^{\infty}\gamma^{k}J_{k}(S[n],\pi)
\label{eq:Vpi}
\end{equation}

The objective of our optimization problem is to maximize the expected distortion reduction: 

\begin{equation}
\max_{\pi} V^{\pi} (S[n]) 
\label{eq:valueiterproblem}
\end{equation}
subject to the following conditions
\begin{enumerate}
\item
$v_{jl}^i > 0$, if for server $i$, we have $r_{jl}^i>0$ ($r_{jl}^i$ is the rank of the packets from the $j$th generation and $l$th layer at server $i$)

\item
$v_{jl}^i = 0$, if $r_{jl} = \beta_l$

\item 
$\mathcal{G}[n]=\{G_{m}: t_{n}<D_{m}\leq t_{n+1}\}$
\end{enumerate}

The first condition means that the receiver requests packets that belong to a given class only if the server has packets of this class or if it can generate them by PRLC encoding. The second condition implies that, once a receiver decodes the layer $l$, it will not request packets from the first $l$ layers as it has  already $\beta_l$ innovative packets from these layers. Hence, any further received packet from classes $l^{'} \le l$ cannot bring innovative information. Finally, the last condition determines the set of active generations. 

The above MDP can be solved by the value iteration method \cite{BellmanValIter57}, which is a form of dynamic programming algorithm. This method updates the values of $V(\cdot)$ function according to the Bellman equation 
\begin{equation}
V_{k+1}(S[n]) =\max\limits_{A[n]} \left\{ \sum\limits_{S[n+1]} \left(\Pr \left( S[n+1]\ \big| \ S[n],A[n]\right) J_{k+1}(S[n],A[n]) + \gamma V_k(S[n+1]) \right) \right\},
\label{eq:valiter}
\end{equation}
where $V_k(\cdot)$ is the value function at the $k$th iteration and $J_{k+1}(S[n],A[n])$ is given in (\ref{eq:objectivefunc}) and the transition probability $\Pr(S[n+1] \big| S[n], A[n])$ is computed as in Appendix \ref{sec:transprob}. This process is repeated until $\big|V_{k+1}(S[n])-V_{k}{S[n]}\big|,\;\forall S[n]\in \mathcal{S}$ is smaller than a threshold value. If the algorithm terminates at the $k$ iteration, the derived policy $\pi$ is given by
\begin{equation}
\pi(S[n])=\arg\max\limits_{A[n]} \left\{ \sum\limits_{S[n+1]} \left(\Pr \left( S[n+1]\ \big| \ S[n],A[n]\right) J_{k+1}(S[n],A[n]) + \gamma V_k(S[n+1]) \right) \right\} \; \forall S[n] \in \mathcal{S}.
\label{eq:optvaliter}
\end{equation}

More details can be found in Section 4.4 of \cite{ReinforcementLearning}.

\section{Q-Learning solution}
\label{sec:overallqlearn}

MDP offers an interesting framework to formalize the scheduling and coding decision processes in multi-source streaming, but it unfortunately requires real time computation of the transition probabilities and rewards functions as shown in (\ref{eq:valiter}). It also necessitates significant storage space, which can be as large as $|\mathcal{S}|\times|\mathcal{A}|\times |\mathcal{S}|$ for the transition probability function. The cost of computing and storing these functions becomes very large as the states space increases exponentially with the number of packets per data layer. In order to avoid the computation of these functions beforehand, we propose alternative solutions to (\ref{eq:valueiterproblem}) based on model-free reinforcement learning approaches \cite{ReinforcementLearning} and more specifically on the Q-learning method\cite{QLearning}, which is a TD (Temporal Difference) control algorithm. Although optimality is not guaranteed with the Q-learning method, the resulting solutions are close to the optimal ones as long as the method runs for a sufficiently large number of iterations. We describe below the Q-learning method, and more specifically the virtual experience algorithm that is used in our solution.

\subsection{Q-Learning Algorithm}
\label{sec:qlearn}

Similarly to the value iteration algorithm, the Q-learning aims to find the optimal policy $\pi^{\star}$ that maximizes the expected reward given in (\ref{eq:Jk}). It achieves that with no need to precompute the transition probabilities and rewards functions and hence can be applied online. Specifically, the Q-learning algorithm \cite{QLearning} computes a $Q(S[n],A[n])$ function for deciding the $\pi^{\star}$. This is a function of both the states and the actions and has size equal to $|\mathcal{S}| \times |\mathcal{A}|$. Q-learning learns the optimal actions to be taken based on experience tuples of the form $(S[n],A[n],J(S[n],A[n]),S[n+1])$. 

In more details the Q-learning algorithm works as follows. Initially, all the entries of the $\mathbf{Q}$ function are set to zero. At the $n$th decision interval (when the algorithm is applied online, the decision interval may coincide with the time $t_n$ of the delivery process), the next state $S[n+1]$ is determined by the current state $S[n]$ and an action $A[n]$ that is selected with probability
\begin{equation}
\frac{e^{Q(S[n],A[n])/\Theta(n)}}{\sum\limits_{A[n+1] \in A} e^{Q(S[n],A[n+1])/\Theta(n)}},
\label{eq:boltz}
\end{equation} 
where $\Theta(n)$ is the temperature value at the $n$th decision interval. For a given action $A[n]$, one can then find the next state $S[n+1]$ and calculate the $J(S[n],A[n])$. The Q-learning algorithm then updates the function value $Q(S[n],A[n])$ as follows 
\begin{equation}
Q(S[n],A[n]) \leftarrow (1-\lambda_n) Q(S[n],A[n]) + \lambda_n (J(S[n],A[n])+ \gamma \max\limits_{A[n+1]} Q(S[n+1],A[n+1])),
\label{eq:Qform}
\end{equation}
where  $\lambda_n$ is the learning rate at time $t_n$ that is given by $$\lambda_n=\frac{1}{1+N_v(S[n],A[n])},$$ where $N_v(S[n],A[n])$ is the number of visits of $(S[n],A[n])$ action-state pair up to the $n$th iteration (time $t_n$). For time $t_1$, the parameter $\lambda_1$ is set to one. 

The above procedure is repeated $\mathcal{N}$ times. The value of $\mathcal{N}$ depends on the considered application. At time $t_n$, the policy $\pi$ ($A[n]=\pi(S[n])$) can be found as
\begin{equation}
\pi(S[n])=\max_{A[n]} Q(S[n],A[n]), \;\forall S[n] \in \mathcal{S}
\end{equation}

As $\mathcal{N}$ goes to infinity, the learned $Q(S[n],A[n])$ function approximates the optimal $Q^{\star}(S[n],A[n])$ and the derived policy is identical to the optimal policy $\pi^{\star}$. This is true when the transition probability is stationary, when all the state-action pairs are visited infinitely often and when $\lambda_n$ satisfies the stochastic approximation conditions, \textit{i.e.}, $\sum_{n=0}^{\infty} \lambda_n=\infty$ and $\sum_{n=0}^{\infty} \lambda_n^2<\infty$ \cite{MastronardeTCOMP2013}. The accuracy of the approximation depends on $\lambda_n$ and $\Theta(n)$ that will be discussed later and the number of Q-learning iterations (episodes) $\mathcal{N}$. 
In practice, the optimal $Q^{\star}(S[n],A[n])$ can generally be found for $\mathcal{N} \ll \infty$.

A critical parameter for the convergence of the Q-learning algorithm is the temperature value $\Theta(n)$ used in (\ref{eq:boltz}). In the beginning, the temperature value $\Theta(n)$ is typically set to a high value in order to permit consistent exploration of the states-actions space. This is called the exploration phase, as all the pairs $(S[n],A[n])$ have almost equal probability of being selected due to the high value of $\Theta(n)$. With time, the temperature value decreases and the knowledge acquired during the exploration phase is exploited, \textit{i.e.}, the values of the $Q(\cdot)$ function are not equal and the probabilities given by (\ref{eq:boltz}) become different. The temperature generally decreases following the model in \cite{BartoAI95}, \textit{i.e.}, 
\begin{equation}
\Theta(n)=\Theta_{\min}+\varphi \cdot (\Theta(n-1)-\Theta_{\min}),
\end{equation}
where $n$ is the index of the decision interval (the index of the current episode), $\Theta(0)=\Theta_{\max}$ is the initial value of the temperature, $\Theta_{min}$ is the minimum temperature and $\varphi$ is a parameter that controls the exploration/exploitation rate. The value of $\varphi$ is found by experimentation and takes values in the range $(0,1]$. 

\subsection{Q-Learning with Virtual Experience}
\label{sec:qlearnVE}

Although Q-learning eliminates the need to precompute the transition probabilities and rewards matrices that is required by the value iteration algorithm (\ref{eq:valiter}), it is characterized by a slow convergence rate, which may limit its application in practical settings. Hence, in order to improve the convergence rate, we propose to use in our problem a variant of the Q-learning algorithm that is hereafter referred to as ``Q-learning VE'', where VE stands for virtual experience. Our method is inspired by \cite{MastronardeTCOMP2013} where the concept of virtual experience has been first proposed for improved convergence of the PDS-learning algorithm (a variant of Q-learning). 

Differently from the Q-learning algorithm, where, in each decision interval, a single state-action pair $Q(S[n],A[n])$ is updated, the ``Q-learning VE'' simultaneously updates multiple state-action pairs $Q(S[n],A[n])$ that are statistically equivalent (virtual states-actions). In particular pair $(\bar{S}[n],\bar{A}[n])$ is statistically equivalent to $(S[n],A[n])$ if the transition probability $\Pr(S[n+1] \big| S[n], A[n])=\Pr(\bar{S}[n+1] \big| \bar{S}[n], \bar{A}[n])$ and the reward $J(\bar{S}[n],\bar{A}[n])$ can be determined from $J(S[n],A[n])$. 

\floatname{algorithm}{Algorithm}
\begin{algorithm}
\caption{Q-learning VE algorithm}
\label{algo:Qlearning_algo}
\begin{algorithmic}[1]
  \STATE Initialize: $\mathbf{Q}=0$, $\mathbf{N_v}=0$, $\Theta(0)=\Theta_{\max}$, $n=1$, and $U$
  \STATE Observe $S[n]$
  \REPEAT
  	\STATE $N_v(S[n],A[n])=N_v(S[n],A[n])+1$
	\STATE $\lambda_n=\frac{1}{1+N_v(S[n],A[n])}$
	\STATE $\Theta(n)=\Theta_{\min}+\varphi \cdot (\Theta(n-1)-\Theta_{\min})$
  	\STATE Select action $A[n]$ with probability $\frac{e^{Q(S[n],A[n])/\Theta(n)}}{\sum\limits_{A[n+1] \in A} e^{Q(S[n],A[n+1])/\Theta(n)}}$
  	\STATE Sample the noise function (\ref{eq:lossprob}) 
	\STATE Find the next state $S[n+1]$ considering $S[n]$, the action A[n] (step 7) and the noise (step 8)
	\STATE Calculate the discounted reward $J(S[n],A[n])$ as in (\ref{eq:Jk})
  	\STATE $Q(S[n],A[n]) \leftarrow (1-\lambda_n) Q(S[n],A[n]) + \lambda_n (J(S[n],A[n])+ \gamma \max\limits_{A[n+1]} Q(S[n+1],A[n+1]))$
	\IF{$\mod(n,U)=0$}
	\FOR{$\bar{S}[n] \in \mathcal{S}$}
	\FOR{$\bar{A}[n] \in \mathcal{A}$}
	\IF{$\Pr(S[n+1] \big| S[n], A[n])=\Pr(\bar{S}[n+1] \big| \bar{S}[n], \bar{A}[n])$, $J(S[n],A[n])=J(\bar{S}[n],\bar{A}[n])$ and $\bar{S}[n+1]=S[n+1]$}
	\STATE $N_v(\bar{S}[n],\bar{A}[n])=N_v(\bar{S}[n],\bar{A}[n])+1$
	\STATE $\lambda_n=\frac{1}{1+N_v(\bar{S}[n],\bar{A}[n])}$
	\STATE $Q(\bar{S}[n],\bar{A}[n]) \leftarrow (1-\lambda_n) Q(\bar{S}[n],\bar{A}[n]) + \lambda_n (J(\bar{S}[n],\bar{A}[n])+ \gamma \max\limits_{A[n+1]} Q(S[n+1],A[n+1]))$
	\ENDIF
	\ENDFOR
	\ENDFOR
	\ENDIF
  	\STATE $S[n] \leftarrow S[n+1]$
	\STATE $n=n+1$
  \UNTIL {$n \le \mathcal{N}$}
  \STATE $\pi(S[n])=\max_{A[n]} Q(S[n],A[n]), \;\forall S[n] \in \mathcal{S}$
\end{algorithmic}
\end{algorithm}

In our specific problem, when an action $A[n] $ is selected, a request allocation vector $v$ is sent to the receivers. Since the links connecting the servers with the receiver are lossy, the vector of packets $\bar{v}$ that arrive at the receiver may be different from the request allocation vector $v$ that is determined by the action $A[n]$. The probability the receiver to observe a vector  $\bar{v}$ while the request allocation vector was $v$ is equal to 
\begin{equation}
\prod\limits_{i=1}^{m} \prod\limits_{j=1}^{L} \prod\limits_{l \in \mathcal{G}_n} {v_{lj}^i \choose \bar{v}_{lj}^i}(\epsilon_{i})^{v_{lj}^i-\bar{v}_{lj}^i}(1-\epsilon_i)^{\bar{v}_{lj}^i},
\label{eq:lossprob}
\end{equation}
where $L$ is the number of data layers per generation, $m$ is the number of receivers, $\mathcal{G}_n$ is the set of active generations at time $t_n$ and  $\epsilon_i$ is loss probability on the link connecting the $i$th sender with the receiver. Therefore, given an action $A[n]$, \textit{i.e.} a request allocation vector $v$, and (\ref{eq:lossprob}), we can find the next state $S[n+1]$ and calculate $J(S[n],A[n])$ as in (\ref{eq:Jk}). 

The convergence is improved with the Q-learning VE, but we should mention that it comes with increased computational complexity as we perform multiple updates of the (\ref{eq:Qform}) for all equivalent states-actions pairs. The frequency of the updates $U$ (every $U$ iterations the algorithm updates all statistically equivalent virtual states-actions), hence constitutes a tradeoff between convergence speed and computational complexity. In each decision interval, we update all the $Q(\bar{S}[n],\bar{A}[n])$ that result in the same reward $J(S[n],A[n])$, \textit{i.e.}, the same expected distortion reduction, and lead to the same next state $S[n+1]$ instead of updating only a single state-action pair $Q(S[n],A[n])$ as is done by the Q-learning algorithm. Moreover, we assume that all the request allocation vectors $v$ (that are determined by the $A[n]$) are affected by the same error pattern. This is due to the fact that the channel conditions are independent from the current state $S[n]$ and the action $A[n]$, \textit{i.e.}, the probability of packet loss depends only on the channel and not to the packet class and generation. We should note that in our system the link capacity is always fully exploited and the receivers never send packet replicas. This becomes possible as PRLC belong to the rateless codes family. 

The proposed Q-learning-VE algorithm is finally summarized in Algorithm \ref{algo:Qlearning_algo}.

\section{Simulation results}
\label{sec:expres}

In this section, the proposed MDP and Q-learning solutions are examined for transmission of layered data from one or multiple servers. We describe first the settings of a transmission scenario with one server. Then, we evaluate the performance of the optimized policies in terms of quality in video transmission from one server and eventually two servers.

\subsection{Data delivery from one server: setup}
\label{sec:OneSerOneC}

We first consider a scenario where the receiver obtains layered data from the server 1 at a rate equal to 1 \textit{packet/time slot}. The $n$th generation of the layered data stream is available at the server at $t=D_{G}n=5n$ and the playback delay $D_{0}$ is 10 time slots. Therefore, the decoding deadlines for the $n$th generation is $D_{n}=10+5n$. The receiver makes decisions at $t=t_{0}=5$ and then every $T=5$ time slots such that $\tau[n]=5$ for all the states $S[n]$. 

In this simple scenario, the state of the server's buffer $B_1[n]$ is deterministic for all the time slots since the data generations are only available  progressively at the server. As a result, we can remove $\tau[n]$ and $B_1[n]$ from the states definition in (\ref{eq:rough}). In order to further simplify the analysis, we remove the packet transmission delays and set $\eta_1=0$. Since the decisions are made at ($t_{1},t_{2},\dots$) and since these time instants coincide with the decoding deadlines $(D_{0},D_{1},\dots)$, at time $t_n$ there is only the decoding matrix $\bm{W}_{n}$ (recall that this matrix contains the innovative packets from generation $n$) in the buffer of the receiver, which should be played out immediately. The matrices that correspond to $\bm{W}_{n-l}, \;l<n$ have been removed as their corresponding deadlines $D_{n-1}, D_{n-2}, \dots$ have already passed, while the $\bm{W}_{n+1}, \bm{W}_{n+2}, \dots$ are empty as at time $t+n$ the layered data for generations $n+l, \; l>n$ is not available at servers. Hence, the states definition contains only the current generation, \textit{i.e.}, $\mathcal{G}[n]=\{G_{n}\}$ and the receiver might only gain a reward (observe a distortion reduction) by decoding $G_{n}$ at time $t_n$. 

The actions are defined as $A[n]=A_1[n]$, as the layered data is transmitted from only one server in the present case. Since $t_{n}<D_{n}$ in every decision time, there are exactly two generations from which packets can be requested, namely $G_{n}$ and $G_{n+1}$ in order to allow foresighted behavior. These packets become available at the server at $t_{n-1}$ and $t_{n}$, respectively. The considered generations depend on the channel capacity and the decoding deadlines. When the channel capacity is larger and the decoding deadlines are more distant, the actions can also involve packets from generations $G_{n+l},\; l>1$, which results in larger actions space.

In the above settings we further have,

\begin{equation}
\Pr \left( S[n+1]\ \big| \ S[n],A[n]\right)=\Pr \left( B[n+1]\ \big| \ B[n],A[n]\right)
\end{equation}
 
The transition probability from a state $S[n]$ to $S[n+1]$ given an action $A[n]$ (the request allocation vector that is sent from the receiver to the server) can be computed as discussed in Appendix \ref{sec:transprob}, where the states are described by the buffer maps $B[n]$ and $B[n+1]$ of the receiver at the decision intervals $n$ and $n+1$. 
 
\subsection{Transmission of scalable video data from one server}
\label{sec:videores}

In this section, we evaluate the policies of the MDP algorithm presented in Section \ref{sec:MDPoptim} which we call it hereafter ``model-MDP'' (this scheme considers the actual loss rates for deciding the optimal scheduling policies) and the Q-Learning algorithms shown in Section \ref{sec:overallqlearn} in an illustrative scenario for transmission of scalable video data. The video data is encoded in two data layers, namely a base and one enhancement layer. The parameters of the source data are reported in Table \ref{tab:gopParam}. All the generations have constant size, \textit{i.e.}, $D_G=5$. 

\begin{table}[h]
\small
\caption{Parameters of a generation with two data layers. The parameters $\alpha_i$, $\beta_i$, $\delta_i$ stand for the number of packets in the $i$th layer, the cumulative number of packets in the first $i$ layers and the distortion reduction after the successful decoding of the $i$th layer, respectively.}
\label{tab:gopParam}
\vspace{-0.5cm}
\begin{center}
\begin{tabular}{|c|c|c|c|}
\hline
 & $\alpha_{i}$ & $\beta_{i}$ & $\delta_{i}$ \\ \hline
Base Layer & 3 & 3 & 11\\ \hline 
Enhancement Layer & 2 & 5 & 9 \\ \hline
\end{tabular} 
\end{center}
\vspace{-0.5cm}
\end{table}

Using the settings presented in Section \ref{sec:OneSerOneC}, the states are described by only the receiver buffer, \textit{i.e.}, $S[n]=\bm{r}_{n}$ for the decoding deadlines presented in \ref{sec:OneSerOneC}. By evaluating the valid rate allocation vectors in (\ref{eq:validrankvectors}), we see that $\bm{r}_{n}$ can take 18 different values in this simple scenario, which are reported in Table \ref{tab:posranks}.

\begin{table}[h]
\small
\caption{States Space. A state is described by a two-tuple which has as entries the number of packet per layer (with increasing order) of the most urgent generation.}
\label{tab:posranks}
\vspace{-0.5cm}
\begin{center}
\begin{tabular}{|c|c|c|c|c|c|c|c|c|}
\hline
(0,0) & (0,1) & (0,2) & (0,3) & (0,4) & (0,5) & (1,0) & (1,1) & (1,2) \\ \hline
(1,3) & (1,4) & (2,0) & (2,1) & (2,2) & (2,3) & (3,0) & (3,1) & (3,2) \\ \hline
\end{tabular}
\end{center}
\vspace{-0.5cm}
\end{table}

For the above setting, the actions involve packets from generations $G_n$ and $G_{n+1}$ and hence there are $M_1=4$ possible packet types. Each action involves the transmission of $N_1=fT=1\times5=5$ packets, as the channel does not experience delay and the decisions are made every 5 time slots. Under these assumptions, it is trivially shown by evaluating (\ref{eq:actionsspace}) that the action space size is $\binom{5+4-1}{5}=56$. A few possible actions in this scenario are presented in Table \ref{tab:actions}.

\begin{table}[h]
\small
\caption{Samples of actions $A[n]$ (Request vectors at $t_{n}$). $(G_n,l)$ packet type corresponds to packets of class $l$ and generation $G_n$. Each row contains the number of packets per class and for generations $G_n$ and $G_{n+1}$.}
\label{tab:actions}
\vspace{-0.5cm}
\begin{center}
\begin{tabular}{|c|c|c|c|}
\hline
$(G_{n},1)$ & $(G_{n},2)$ & $(G_{n+1},1)$ & $(G_{n+1},2)$ \\ \hline \hline
3 & 2 & 0 & 0 \\ \hline
0 & 0 & 5 & 0 \\ \hline
0 & 2 & 2 & 1 \\ \hline
2 & 1 & 1 & 1 \\ \hline
\end{tabular}
\end{center}
\vspace{-0.5cm}
\end{table}

The first row of Table \ref{tab:actions} corresponds to an action (and a request allocation vector) where the receiver demands from the server to send three packets from the first class of generation $G_n$ and two packets from the second class of the same generation, while no packets that belong to the generation $G_{n+1}$ are requested from the server.

Before proceeding with the evaluation, we should note that here the policies for the proposed schemes are determined offline. Then, the derived policies are evaluated in the simulations for various error patterns. We follow the value iteration method \cite{BellmanValIter57} in order to solve the proposed model-MDP. The threshold value is set to $10^{-9}$. When this threshold value is reached, the output of the value-iteration algorithm is the policy used in the simulations. This policy is considered as the optimal and is used as a benchmark. Similarly, the Q-learning algorithm outputs the tested policy when the maximum number of iterations $\mathcal{N}$ is reached. In all the experiments, the parameters $\Theta(0)$ and $\Theta_{\min}$ of the Q-learning algorithms are set to 75 and 0.5, respectively. All the presented results are averaged over 100 simulations. Each simulation has duration of 100 sec, namely 100 generations. For the sake of fairness, in all schemes under comparison, the transmitted packets in each link experience the same loss patterns.

\subsubsection{Foresighted versus Myopic policies}
\begin{figure}[t]
\centering
\includegraphics[width=0.5\textwidth]{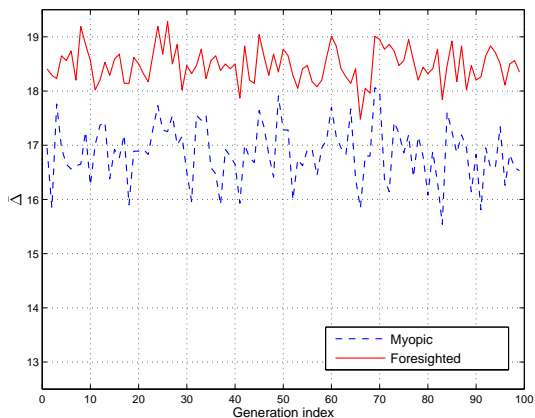}
\vspace{-0.5cm}
\caption{Average distortion reduction evaluation of the model-MDP algorithm optimized for $\gamma=0.0$ (myopic) and $\gamma=0.9$ (foresighted) with respect to the generation index. The  scalable video is encoded in two quality layers and the video parameters are listed in Table \ref{tab:gopParam}. The link connecting the server with the receiver experiences 5\% loss rate.}
\vspace{-0.5cm}
\label{fig:MyopicVSForeSighted}
\end{figure}

We consider that the scalable video is encoded in two quality layers. The video parameters are reported in Table \ref{tab:gopParam}. The link connecting the server with the receiver has 5\% loss rate. We examine the influence of the discount factor $\gamma$ for the Model-MDP scheme to the quality (\textit{i.e.}, the achieved average distortion reduction $\overline{\Delta}$). Specifically, the value iteration algorithm is executed for $\gamma=0.9$ and $\gamma=0.0$, which correspond to foresighted and myopic policies, respectively. The average cumulative distortion reduction with respect to the generations index is depicted in Fig. \ref{fig:MyopicVSForeSighted}. We observe that the foresighted policies offer large performance gains (average gain of approximately 1.5 dB) compared to the myopic policies. Further, we can see that the foresighted policies offer small quality fluctuations compared to the myopic ones which is important for streaming systems. Hence, the joint consideration of multiple generations in the scheduling and coding decisions is certainly advantageous. Therefore, in the following we focus only on foresighted policies. 

\subsubsection{Influence of the number of Q-learning algorithms iterations to the quality}

We evaluate the performance of Q-learning and Q-learning VE algorithms in terms of the average expected distortion reduction $\overline{\Delta}$ for various numbers of iterations $\mathcal{N}$. The $\varphi$ values are found by experimentation for each $\mathcal{N}$.\footnote{The same $\varphi$ values are considered for both Q-learning algorithms. These are not presented for brevity reasons.} We assume that the frequency of the updates $U$ of the virtual states-actions is 10 (every 10 decision intervals, we apply the batch update, \textit{i.e.}, we perform the steps 12-20 of Algorithm \ref{algo:Qlearning_algo}, for all equivalent states-actions pairs). The results are presented in Fig. \ref{fig:IterationsVsPSNR}. For the sake of completeness, we depict the average value of the distortion reduction achieved by the Model-MDP for $\gamma=0.0$ and $\gamma=0.9$. From Fig. \ref{fig:IterationsVsPSNR}, we can observe that Q-learning VE requires a relative small number of iterations in order to reach the performance of Model-MDP with $\gamma=0.9$ (foresighted). Specifically, for $\mathcal{N}=50000$, the performance difference between Q-learning VE and the Model-MDP with $\gamma=0.9$ is only 0.06 dB. We can also see that the Q-learning algorithm requires more than 200000 iterations to converge to the performance of the foresighted Model-MDP. We should emphasize that Q-learning VE algorithm achieves significantly better performance than Q-learning for the same number of iterations with slightly higher computational complexity, as the batch update of all equivalent virtual states-actions is performed every 10 iterations.

\begin{figure}[t]
\centering
\includegraphics[width=0.5\textwidth]{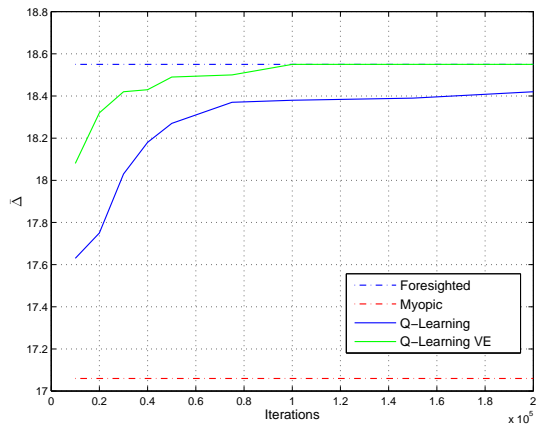}
\vspace{-0.5cm}
\caption{Average cumulative distortion reduction achieved by the Q-learning and Q-learning VE algorithms, optimized for $\gamma=0.9$, with respect to the number of iterations. The distortion values of Model-MDP algorithm for $\gamma=0.0$ (myopic) and $\gamma=0.9$ (foresighted) are also depicted. Transmission of a video encoded in two quality layers from a single server to a receiver with the parameters listed in Section \ref{sec:OneSerOneC} is considered.}
\vspace{-0.5cm}
\label{fig:IterationsVsPSNR}
\end{figure}

\subsubsection{Influence of the update rate $U$ of the Q-learning VE algorithm to the quality}

We examine the influence of the frequency of the updates $U$ of the Q-learning VE algorithm on the average expected distortion reduction. The $\varphi$ is set to 0.99986. The results are illustrated in Fig. \ref{fig:PSNRvsUpdate}. We consider 50000 iterations and then we change the frequency $U$ from 1 (\textit{i.e.}, in every decision interval, we apply the batch update of all the virtual states-actions) up to 100 (\textit{i.e.}, every 100 decision intervals, we perform a batch update of all virtual states-actions). We also depict the performance of the Model-MDP and the original Q-learning algorithm with $\gamma=0.9$. We see that, when the update rate $U$ is 1, the performance difference is less than 0.06 dB, which is slightly better than the performance achieved when $U$ equals to 10 and $\mathcal{N}$ is 200000 (Fig. \ref{fig:IterationsVsPSNR}). As expected, when the updates are less frequent (larger $U$), the performance drops. However, it is worth to mention that the Q-learning VE outperforms significantly the Q-learning even with an update rate of 100. Such an update rate brings only a minimal additional computational cost compared to the Q-learning. This is particularly valuable for scenarios involving receivers with low computational capabilities. 

\begin{figure}[t]
\centering
\includegraphics[width=0.5\textwidth]{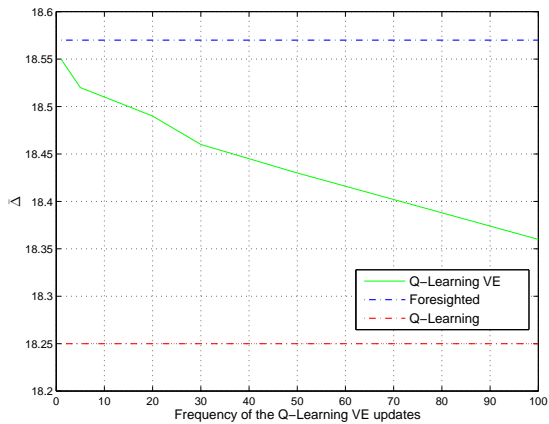}
\vspace{-0.5cm}
\caption{Average cumulative distortion reduction achieved by the Q-learning VE algorithm, optimized for $\gamma=0.9$ and $\mathcal{N}=50000$ iterations, with respect to the frequency of the updates of the virtual actions-states. The distortion values of Model-MDP algorithm for $\gamma=0.9$ and the Q-learning algorithm are also depicted. Transmission of a video encoded in two quality layers from a single server to a receive with the parameters listed in Section \ref{sec:OneSerOneC} is considered.}
\vspace{-0.5cm}
\label{fig:PSNRvsUpdate}
\end{figure}

\subsubsection{Average distortion reduction with respect to the time}

We also examine the model-MDP and Q-learning solutions for a discount factor $\gamma = 0.9$. We consider an update rate $U$ of 10 for the Q-learning VE approach. We set the $(\mathcal{N},\varphi)$ tuple to (250000,0.99996) and (50000,0.99986) for the Q-learning and Q-learning VE schemes, respectively. These solutions are compared with a scheme called ``RandSched''  that implements random scheduling of PRLC packets that belong to the most urgent generations. Specifically, this scheme in each decision interval $t_n$ randomly selects for transmission packets from the generations $G_n$ and $G_{n+1}$. The simulation results are presented in Figs. \ref{fig:time_results_various_loss} (a), (b) and correspond to  5\% and 10\% loss rate, respectively. In all the comparisons, the model-MDP and Q-learning schemes are optimized considering the expected loss rate. We see that model-MDP and the Q-learning solutions perform equally well. We can also notice that Q-learning VE achieves similar performance to Q-learning and that it requires a reduced number of iterations $\mathcal{N}$. Moreover, we note that these solutions outperform significantly the RandSched scheme. This is attributed to the fact that RandSched is oblivious to the data importance and timing constraints. Interestingly, we can further observe in Fig. \ref{fig:time_results_various_loss} (a) that, for a 5\% loss rate model-MDP and Q-learning behave similarly, although their policies are not necessarily the same. This is due to the fact that the requested packets results in similar distortion reduction.

\begin{figure}[t!]
\begin{center}
			\subfloat[5\% loss rate]{\label{fig:loss5}\includegraphics[width=0.5\textwidth]{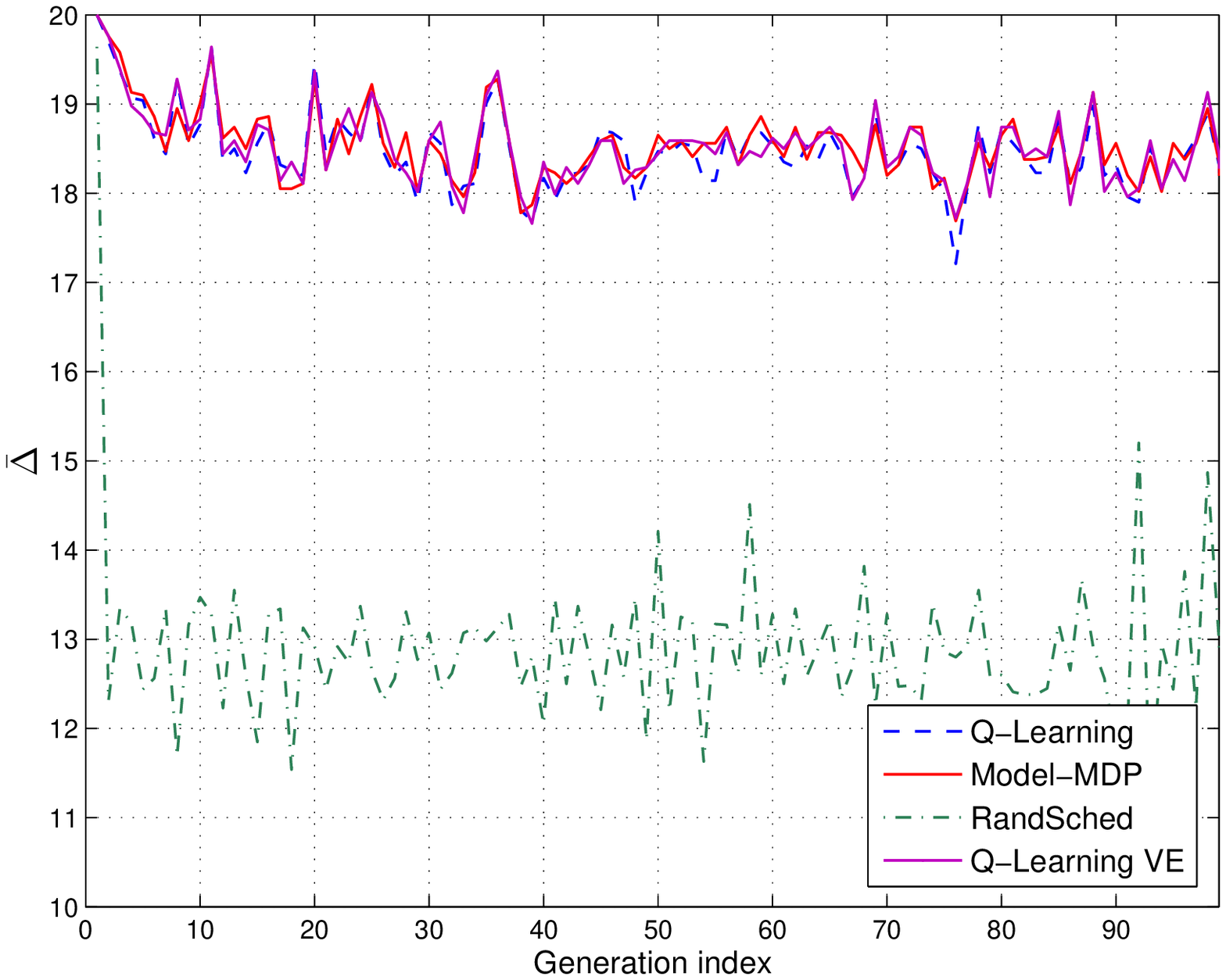}} 
			\subfloat[10\% loss rate]{\label{fig:loss10}\includegraphics[width=0.5\textwidth]{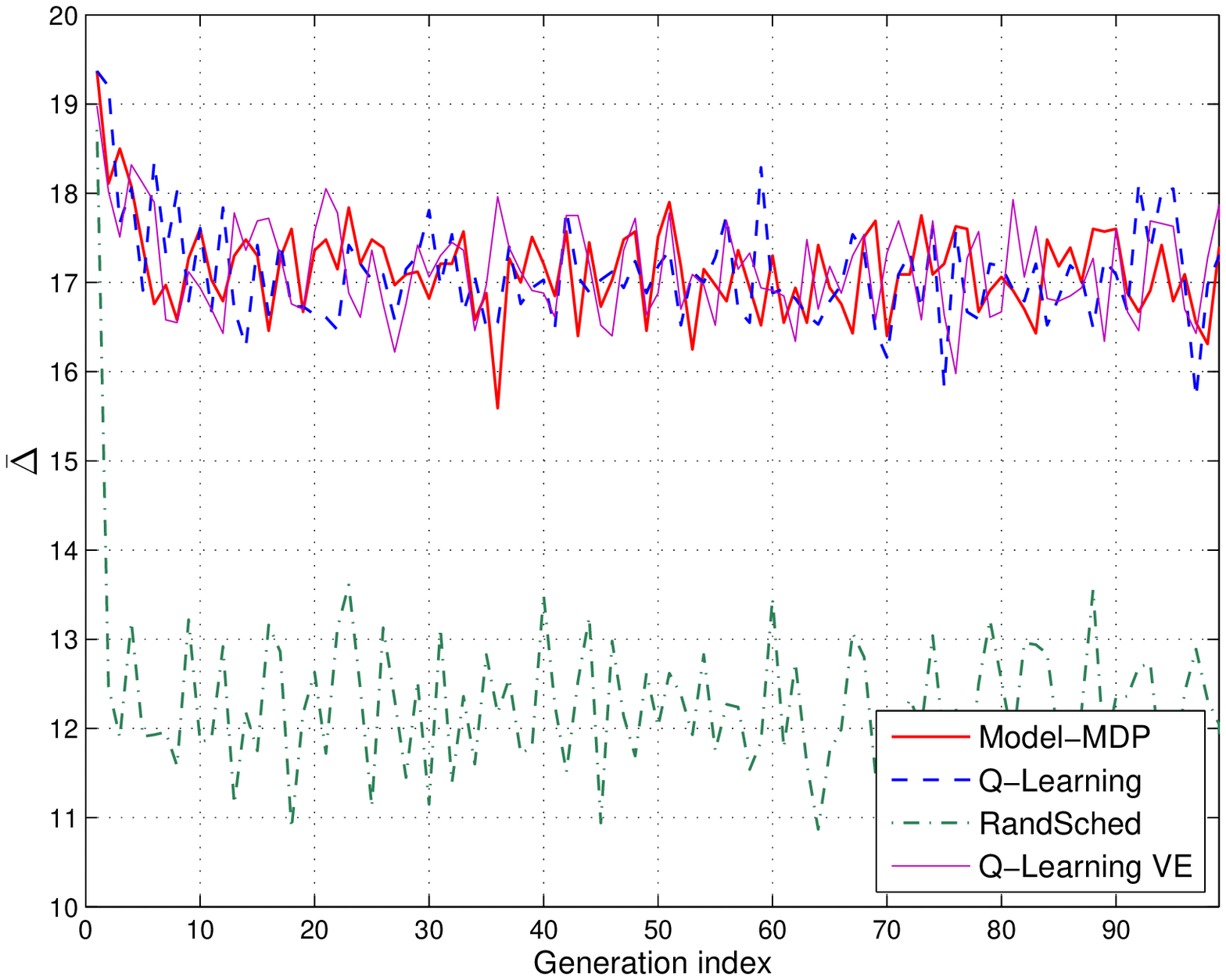}}			
\end{center}
\vspace{-0.5cm}
\caption{Average cumulative distortion reduction $\overline{\Delta}$ for model-MDP, Q-learning, Q-learning VE and Random Scheduling policies with respect to the generation index for transmission of scalable video encoded in two quality layers. The policies are optimized for $\gamma=0.9$ and the loss rates are equal to (a) 5\% and (b) 10\%. The average distortion values for the whole video sequence for the model-MDP, Q-learning VE, Q-learning and RandSched algorithms are 18.55, 18.53, 18.48, and 12.92dB, respectively for 5\% loss rate. When loss rate increases to 10\% the $\overline{\Delta}$ values become 17.16, 17.16, 17.10 and 12.30 dB.}
\vspace{-0.5cm}
\label{fig:time_results_various_loss}
\end{figure}

\subsection{Transmission of scalable video data with three layers from one server }

\begin{table}[t]
\small
\caption{Parameters of a generation of a scalable video encoded in three layers. The parameters $\alpha_i$, $\beta_i$, $\delta_i$ stand for the number of packets in the $i$th layer, the cumulative number of packets in the first $i$ layers and the distortion reduction after the successful decoding of the $i$th layer, respectively.}
\label{tab:gopParam_3layers}
\vspace{-0.5cm}
\begin{center}
\begin{tabular}{|c|c|c|c|}
\hline
 & $\alpha_{i}$ & $\beta_{i}$ & $\delta_{i}$ \\ \hline
Base Layer & 3 & 3 & 11\\ \hline 
Enhancement Layer 1& 2 & 5 & 9 \\ \hline
Enhancement Layer 2& 2 & 7 & 12 \\ \hline
\end{tabular} 
\end{center}
\vspace{-0.5cm}
\end{table}

\begin{figure}[t!]
\begin{center}
			\subfloat[$D_G=5$ time slots]{\label{fig:Dg5}\includegraphics[width=0.5\textwidth]{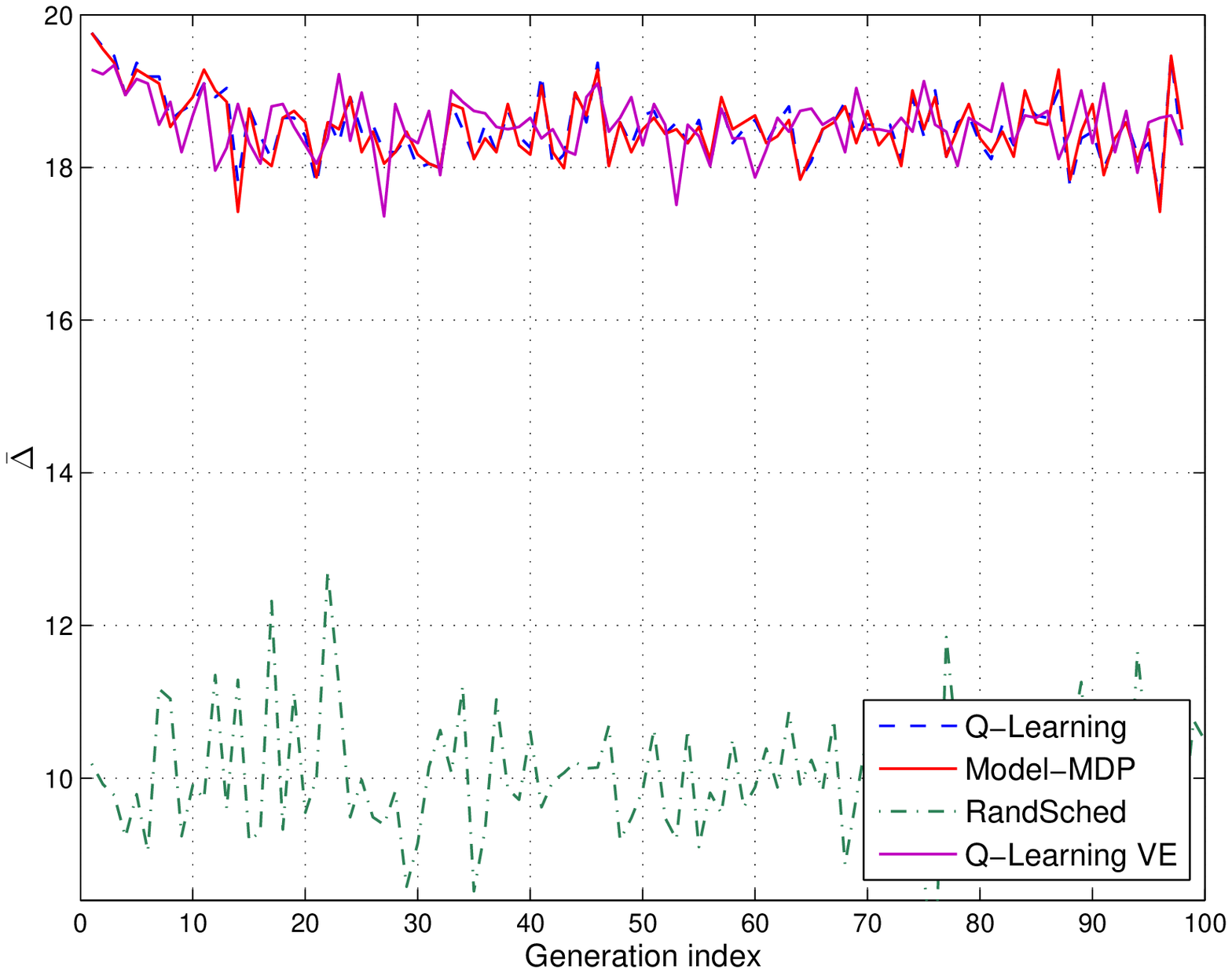}} 
			\subfloat[$D_G=7$ time slots]{\label{fig:Dg7}\includegraphics[width=0.5\textwidth]{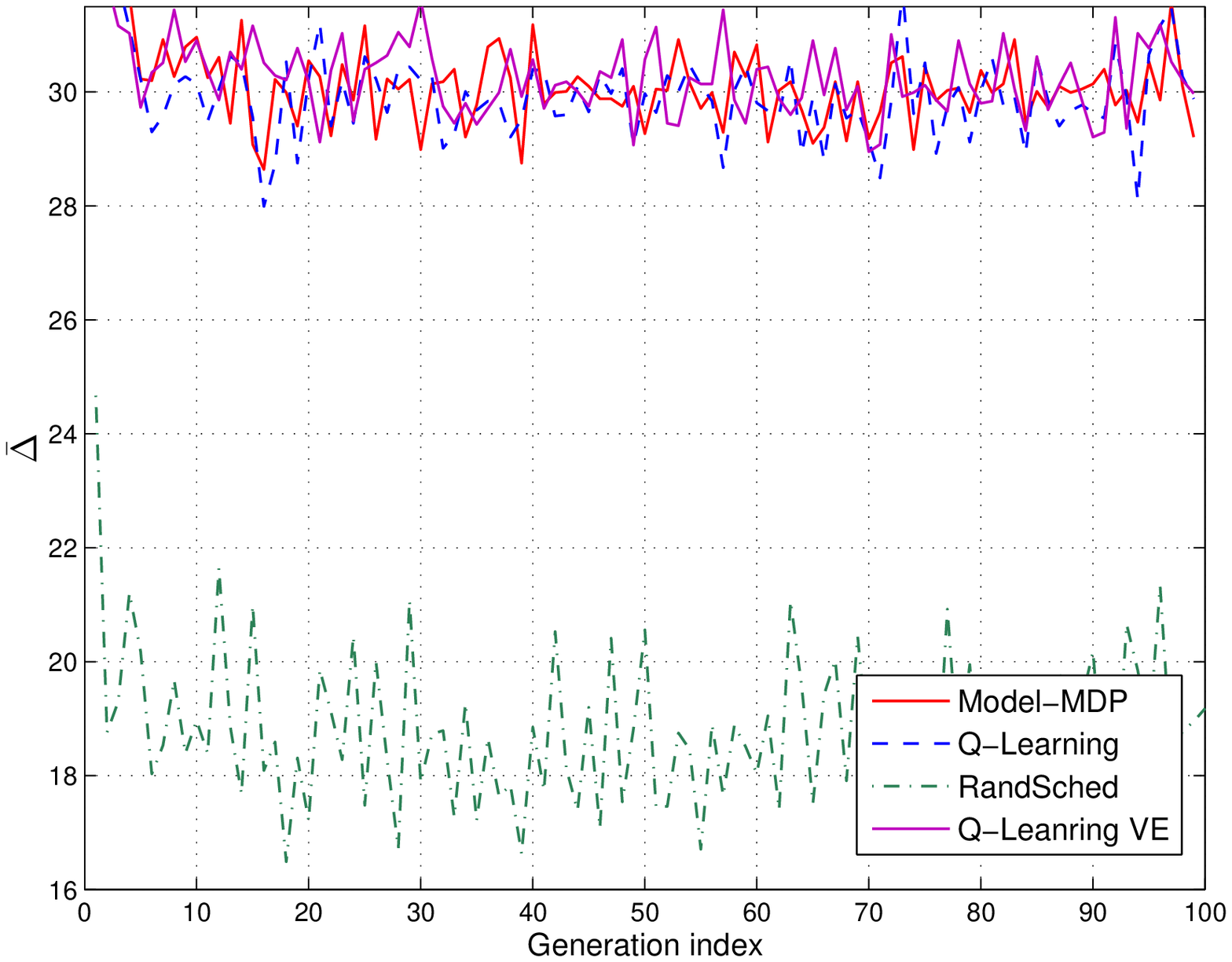}}
\end{center}
\vspace{-0.5cm}
\caption{Average cumulative distortion reduction $\overline{\Delta}$ comparison of model-MDP, Q-learning, Q-learning VE and Random Scheduling policies optimized for $\gamma=0.9$ and 5\% loss rate. The scalable video is encoded in three quality layers and the video parameters are listed in Table \ref{tab:gopParam_3layers}. (a) decision interval is $D_G=5$ time slots (the bandwidth is sufficient to transmit 5 packets in each interval) and (b) the decision interval is $D_G=7$ time slots. When $D_G=5$, the average distortion values for the video sequence are 18.56, 18.53, 18.54, and 10.01 dB for the model-MDP, Q-learning VE, Q-learning and RandSched algorithms, respectively.  When $D_G=7$ the distortion values become 30.28, 3018, 29.95, and 18.81dB.}
\vspace{-0.5cm}
\label{fig:av_PSNR_3layers}
\end{figure}

We further consider the transmission of videos encoded in three quality layers. The parameters of the video are presented in Table \ref{tab:gopParam_3layers}. Similarly to the case with two data layers, the decision interval is set to 5 time slots, which is sufficient for transmitting 5 packets. The $n$th generation of the layered data stream becomes available at the server at $t=D_{G}n=5n$ and the playback delay $D_{0}$ is 10 decision intervals. The decoding deadlines for the $n$th generation is $D_{n}=10+5n$. In the $n$th decision interval, the generations $G_n$ and $G_{n+1}$ are considered as urgent. The model-MDP, the Q-learning and the Q-learning VE algorithms are optimized for $\gamma=0.9 $ and 5\% loss rate. For the Q-learning VE an update rate $U$ of 10 is chosen. We set the $(\mathcal{N},\varphi)$ tuple to ($2 \cdot 10^6$,0.999995) and ($2 \cdot 10^5$,0.99996) for the Q-learning and Q-learning VE schemes, respectively. All the above schemes are compared with the RandSched algorithm and the results are given in Fig. \ref{fig:av_PSNR_3layers}(a). As expected, model-MDP and Q-Learning solutions perform similarly in terms of distortion reduction to the case of video encoded in two layers (Fig. (\ref{fig:loss5})). This is due to the fact that the time is not sufficient for transmitting packets containing data from the third layer. From Fig. \ref{fig:av_PSNR_3layers}(a), we can also observe that the RandSched scheme performs very poorly. The performance of the RandSched scheme is lower than its performance when the video is encoded in two layers (Fig. \ref{fig:loss5}). Despite the limited bandwidth, packets that belong to third class are transmitted with RandSched. Therefore, the probability to decode the second layer becomes smaller. We remark that the Q-learning VE is able to achieve the same performance with the Model-MDP with 10 times less iterations compared to those required by the Q-learning algorithm. 

Finally, we evaluate the same schemes when a larger bandwidth is available. Specifically, we assume that the duration of each time slot is shorter, namely $1/7$ sec, hence the bandwidth is 7 \textit{packets/sec}. The decision interval is set to 7 time slots, which is sufficient for transmitting 7 packets. The decoding deadlines for the $n$th generation is $D_{n}=10+7n$. Again, two generations are marked as urgent. We set the $(\mathcal{N},\varphi)$ tuple to ($13 \cdot 10^6$,0.9999986) and ($4 \cdot 10^5$,0.999983) for the Q-learning and Q-learning VE schemes, respectively. The results are illustrated in Fig. \ref{fig:av_PSNR_3layers}(b). We can see that the model-MDP and Q-Learning algorithms are able to take advantage of the additional bandwidth and improve significantly the video quality. We can further observe that both model-MDP and Q-learning VE outperform Q-learning. This performance difference is due to the fact that the states space grows from 18 to 88 and the actions space from 56 to 792. Specifically, as the states and actions space becomes larger, the probability that Q-learning stops in a local minimum increases. Identical performance of model-MDP and Q-Leaning is still possible by properly selecting the $\varphi$ and $\mathcal{N}$ values. From Fig. \ref{fig:av_PSNR_3layers}(b), we can also see that RandSched algorithm is not competitive with the model-MDP and Q-Learning schemes and their performance difference in distortion is approximately 11 dB in this illustrative example. 

\subsection{Video Transmission from Multiple Servers}

\begin{figure}[t]
\centering
\includegraphics[width=0.5\textwidth]{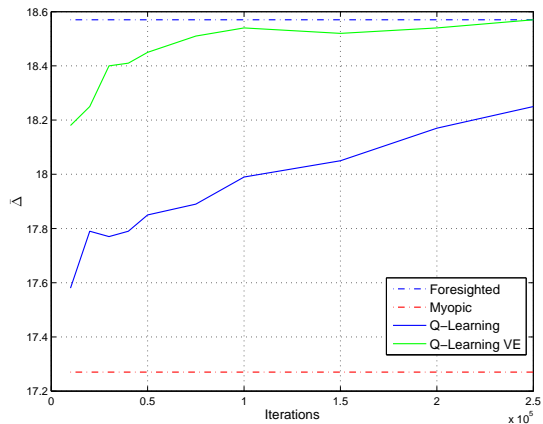}
\vspace{-0.5cm}
\caption{Average cumulative distortion reduction achieved by the Q-learning and Q-learning VE algorithms, optimized for $\gamma=0.9$, with respect to the number of iterations. The distortion values of Model-MDP algorithm for $\gamma=0.0$ (Myopic) and $\gamma=0.9$ (Foresighted) are also depicted. Transmission of a video encoded in two quality layers from two servers to a single receiver with the video parameters listed in Table \ref{tab:gopParam_3layers} is considered.}
\vspace{-0.5cm}
\label{fig:Multiservers_IterationsvsPSNR}
\end{figure}

\begin{figure}[t]
\centering
\includegraphics[width=0.5\textwidth]{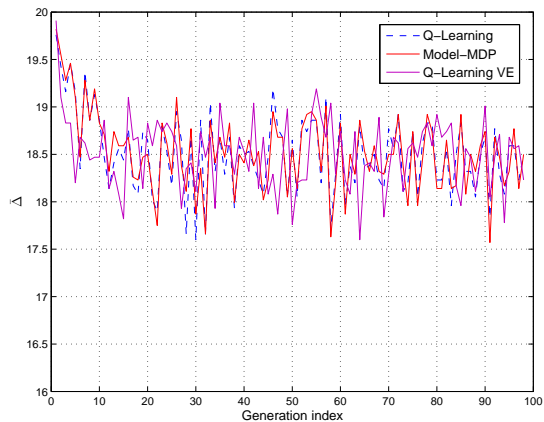}
\vspace{-0.5cm}
\caption{Average cumulative distortion reduction $\overline{\Delta}$ comparison of model-MDP, Q-learning and Q-learning VE policies for transmission of layered video from two servers to a single receiver. The discount factor is $\gamma=0.9$ and the links have loss rates $\epsilon_{1}=15\%$ and $\epsilon_{2}=5\%$ . The scalable video is encoded in two quality layers and decision interval is $D_G=5$ secs. The average distortion values for the video sequence for the model-MDP, Q-learning VE, and Q-learning are 18.52, 18.51, and 18.51dB, respectively.}
\vspace{-0.5cm}
\label{fig:assymetric}
\end{figure}

We finally evaluate our model-MDP and Q-Learning algorithms for  an illustrative scenario of multi-server scalable video transmission. We consider a scenario with two servers with bandwidths $f_{1}=0.6$ \textit{packets/slot} (3 \textit{packets/sec}) and $f_{2}=0.4$ \textit{packets/slot} (2 \textit{packets/sec}) and that both links have a $5\%$ loss rate. Essentially, the examined scenario and the scenario presented in Section \ref{sec:OneSerOneC} are identical from receiver's point of view as the overall capacity is one \textit{packet/sec} and loss rate in both cases is 5\%. First, we examine the performance of Q-learning and Q-learning VE algorithms in terms of the average expected distortion $\overline{\Delta}$ for various numbers of iterations $\mathcal{N}$. The discount factor is set to 0.9. The frequency of the updates of the virtual states-actions $U$ is set to 10. The $\varphi$ value are found by experimentation for each $\mathcal{N}$. We also present the performance of the foresighted and myopic schemes, \textit{i.e.}, Model-MDP for $\gamma=0.0$ and $\gamma=0.9$, respectively. The results are presented in Fig. \ref{fig:Multiservers_IterationsvsPSNR}. We see that Q-learning VE performs significantly better than the Q-learning algorithm in terms of average expected distortion. For 50000 iterations the  performance gap in terms of distortion between the Q-learning VE and the Model-MDP is small and for 100000 iterations it becomes negligible. Further, we can see that Q-learning can not reach the performance (in terms of $\overline{\Delta}$) achieved by the Q-learning VE even with 250000 iterations, and the performance difference is close to 0.3 dB. By observing Figs. \ref{fig:IterationsVsPSNR} and \ref{fig:Multiservers_IterationsvsPSNR}, we can conclude that the increased number of servers affects mainly the performance of the Q-learning algorithm, while the performance of the Q-learning VE stays almost invariant. This is due to the larger states-actions space (the action space increases from 56 to 200, while the action space remains the same with 18 states) that slow down the convergence of the Q-learning algorithm. Q-learning VE does not suffer from that as the batch updates help to exploit the knowledge acquired by visiting an action-state pair to multiple.

Next, we investigate the performance of the model-MDP and Q-Learning algorithms when the loss rate on the link connecting the first server with the receiver increases to 15\%. Hence, the average loss rate becomes $11\%$. The $(\mathcal{N},\varphi)$ tuple is set to ($6.5 \cdot 10^5$,0.999988) and  ($2 \cdot 10^5$,0.99996) for the Q-learning and Q-learning VE, respectively. The results are illustrated in Fig. \ref{fig:assymetric}. We can observe that the model-MDP and Q-learning schemes have the same performance. Like before, the small differences in the curves are due to the existence of multiple policies that are identical in terms of $\overline{\Delta}$. If we further compare the above results with the ones shown in Fig. \ref{fig:loss10} where the receiver gets the data from a single server through a link that has 10\% loss rate, we see that the existence of a more reliable link helps the model-MDP and Q-learning schemes to have improved performance (\textit{i.e.}, the average distortion reduction is larger) even if the average loss rate is larger. This happens as the most important packets are transmitted on the more reliable link.

\section{Conclusions}
\label{sec:conclusions}

We have investigated the problem of jointly selecting the optimal coding strategy and scheduling decisions when receivers can obtain delay constrained layered data encoded with Prioritized Random Linear Codes from multiple servers. Markov Decision Processes are used for formulating the optimization problem. Then, we propose a Q-learning variant called Q-learning VE that exploits the fact that multiple states-actions are statistically equivalent. This permits to find effective coding and scheduling solutions with less iterations compared to the traditional Q-learning method with only slightly higher computational complexity. In illustrative video transmission examples, first we show the benefits of the foresighted over the myopic decisions. The results confirm the advantages of considering multiple generations concurrently in the scheduling. Then, we see that the Q-learning VE algorithm outperforms significantly the Q-learning method in terms of expected distortion reduction for the same number of iterations. We note that for the Q-learning VE algorithms a relative small number of iterations are needed for taking most of the performance gains. Further, we observe that the Q-learning algorithm converges fast to that of the MDP in terms of distortion reduction. Interestingly, with the proposed methods continuous playback and small quality variations are favored. This work is a step towards using PRLC methods for real-time video streaming. Our future research will investigate further simplification of the Q-learning approaches and eventually end up with the definition of simple decision rules that approximate the performance of the Q-learning VE algorithm.

\appendices 

\section{Remarks on Vector Spaces over Finite Fields}
\label{app}

\begin{proposition}
Let $V$ be a vector space of dimension $k$ over a finite field $\mathbb{F}_{q}$ of size $q$ and $N$ be the number of vectors that are drawn from $V$. The probability that the vector space spans a subspace of dimension $r$ is $\phi_{q}(r,N,k)=\qbinom{k}{r}\sum_{i=0}^{r}(-1)^{r-i}\qbinom{r}{i}q^{Ni+\binom{r-i}{2}-Nk}$.
\end{proposition}

\begin{proof}
This is a direct application of combinatorics theory \cite{combinatorics}. Thus, the probability that $\dim U = r$ is

\begin{equation}\label{eq:apq1Phi}
\phi_{q}(r,N,k)=\qbinom{k}{r}\sum_{i=0}^{r}(-1)^{r-i}\qbinom{r}{i}q^{Ni+\binom{r-i}{2}-Nk},
\end{equation}

where $\qbinom{A}{B}$ is a Gaussian coefficient defined as

\begin{equation}
\qbinom{A}{B}=
\begin{cases}
1& k=0,\\
\frac{(q^{A}-1)(q^{A-1}-1)\dots(q^{A-B+1}-1)}{(q^{B}-1)(q^{B-1}-1)\dots(q-1)}& k>0.
\end{cases}
\end{equation}
\end{proof}

\begin{proposition}
Let $V$ be a vector space of dimension $k$ over a finite field $\mathbb{F}_{q}$ of dimension $q$ and $U$ and $W$ two random subspaces with dimensions $n\leq k$ and $m\leq k$, respectively. Then, the probability that the span of the intersection of the subspaces $U$ and $W$ is equal to $r$ is given by $$\Pr(\dim(U\cap W)=r)=\frac{q^{n(m-r)}\qbinom{k-n}{m-r}\qbinom{n}{r}}{\qbinom{k}{m}},$$
where $\dim U = r$.
\end{proposition}

\begin{proof}

We have that $r$ lies in the interval 

\begin{equation}
(n+m-k)^{+}\leq r\leq\min(n,m).
\end{equation}

For any valid $r$, let us enumerate the $W$ subspaces such that $\dim(U\cap W)=r$. Any such $W$ has $m$ basis vectors and $m-r$ of them do not belong to $U$. In other words, $W$ can be decomposed into $W\cap U$ and $W\setminus U$, which are non-intersecting. 

Now, we know that there are $\qbinom{n}{r}$ ways that $W\cap U$ can be selected. We can select the basis vectors $W\cap U$ in $(q^{k}-q^{n})(q^{k}-q^{n+1})\dots(q^{k}-q^{n+m-r-1})$ ways. However, in this selection, exactly $(q^{m-r}-1)(q^{m-r}-q)\dots(q^{m-r}-q^{m-r-1})$ basis vector sets span the same subspace. Therefore, the number of different sets of basis vectors for $W\cap U$ is given by 

\begin{equation}
\frac{(q^{k}-q^{n})(q^{k}-q^{n+1})\dots(q^{k}-q^{n+m-r-1})}{(q^{m-r}-1)(q^{m-r}-q)\dots(q^{m-r}-q^{m-r-1})}\qbinom{n}{r}=
q^{n(m-r)}\qbinom{k-n}{m-r}\qbinom{n}{r}
\end{equation}

There are $W$ subspaces such that $\dim(U\cap W)=r$ holds. Since there are $\qbinom{k}{m}$ $m$-dimensional subspaces in $V$, we have

\begin{equation}
\Pr(\dim(U\cap W)=r)=\frac{q^{n(m-r)}\qbinom{k-n}{m-r}\qbinom{n}{r}}{\qbinom{k}{m}}.
\end{equation}

Equivalently,

\begin{equation}\label{eq:apeq2Final}
\Pr(\dim(U\cup W)=s)=\frac{q^{n(s-n)}\qbinom{k-n}{s-n}\qbinom{n}{n+m-s}}{\qbinom{k}{m}}.
\end{equation}

\end{proof}

\begin{proposition}
Let $V$ be a $k$-dimensional vector space over a finite field $\mathbb{F}_{q}$ of size $q$ and $U_{1}$ \& $U_{2}$ be two subspaces of $V$ with dimensions $m_{1}$ and $m_{2}$, respectively. We randomly draw $N_{1}$ and $N_{2}$ vectors out of $U_{1}$ and $U_{2}$, which span $S_{1}$ and $S_{2}$. Then, the probability of having $s$ independent vectors reads: $$\Pr(\dim(S_{1}\cup S_{2})=s)=\sum_{n_{1},n_{2}}\left(\frac{q^{n_{1}(s-n_{1})}\qbinom{k-n_{1}}{s-n_{1}}\qbinom{n_{1}}{n_{1}+n_{2}-s}}{\qbinom{k}{n_{2}}}\prod_{i=1}^{2}\phi_{q}(n_{i},N_{i},m_{i})\right).$$
\end{proposition}

\begin{proof}

From Eq.\eqref{eq:apq1Phi}, the random vectors drawn from $U_{i}$ span an $n_{i}$-dimensional subspace with probability $\phi_{q}(n_{i},N_{i},m_{i})$. As the vector drawings from each subspace are independent, the resulting subspaces have dimensions $n_{1}$ and $n_{2}$ with probability $\prod_{i=1}^{2}\phi_{q}(n_{i},N_{i},m_{i})$. Then, using Eq.\eqref{eq:apeq2Final} we have

\begin{equation}\label{eq:apq3}
\begin{split}
\Pr(\dim(S_{1}\cup S_{2})=s)&=\sum_{n_{1},n_{2}}\Pr\left(\dim(S_{1}\cup S_{2})=s \ | \ \dim(S_{1})=n_{1},\dim(S_{2})=n_{2}\right)\Pr(n_{1},n_{2}) \\
&=\sum_{n_{1},n_{2}}\left(\frac{q^{n_{1}(s-n_{1})}\qbinom{k-n_{1}}{s-n_{1}}\qbinom{n_{1}}{n_{1}+n_{2}-s}}{\qbinom{k}{n_{2}}}\prod_{i=1}^{2}\phi_{q}(n_{i},N_{i},m_{i})\right).
\end{split}
\end{equation}

\end{proof}

\begin{proposition}
Let $V$ be a vector space of dimension $k$ over a finite field $\mathbb{F}_{q}$ of size $q$ and $\{U_{i}\}$ be the set of $R$ subspaces with $\dim(U_{i})=m_{i}$. We randomly draw $N_{i}$ vectors out of $U_{i}$, which span $S_{i}$. Then, the probability of having $s_R$ independent vectors can be expressed as $$\Pr\left(\dim(\bigcup_{i=1}^{R}S_{i})=s_{R}\right)=\sum_{n_{1},\dots,n_{R}} \sum_{s_{2},\dots,s_{R-1}}\psi_{1}(s_{2}; n_{1},n_{2})\prod_{i=2}^{R-1}\psi_{i}(s_{i+1}; s_{i},n_{i+1}) \prod_{i=1}^{R}\phi_{q}(n_{i}, N_{i}, m_{i}.)$$
\end{proposition}
\label{prop:p3}

\begin{proof}
The following equation is rather straightforward.

\begin{equation}
\Pr\left(\dim(\bigcup_{i=1}^{R}S_{i})=s_{R}\right)=\sum_{n_{1},\dots,n_{R}}\Pr\left(\dim(\bigcup_{i=1}^{R}S_{i})=s_{R}\bigg|\dim(S_{i})=n_{i},\forall i\right) \Pr\big(\dim(S_{i})=n_{i},\forall i\big)
\end{equation}

To expand the first term, let us use the following notation: 

\begin{equation}
\psi_{R}(x; y,z)\triangleq\Pr\left\{\dim\left[\left(\bigcup_{i=1}^{R}S_{i}\right)\cup S_{R+1}\right]=x \ \bigg| \ \dim\left(\bigcup_{i=1}^{R}S_{i}\right)=y ,\dim(S_{R+1})=z\right\}
\end{equation}

Note that $\psi_{1}(s; n,m)$ results in Eq.\eqref{eq:apeq2Final}. Then, the following equalities apply

\begin{equation}
\Pr\left(\dim(\bigcup_{i=1}^{R}S_{i})=s_{R}\bigg|\dim(S_{i})=n_{i},\forall i\right)  = \sum_{s_{2},\dots,s_{R-1}}\psi_{1}(s_{2}; n_{1},n_{2})\prod_{i=2}^{R-1}\psi_{i}(s_{i+1}; s_{i},n_{i+1}).
\end{equation}
and 
\begin{equation}
\Pr\big(\dim(S_{i})=n_{i},\forall i\big) = \prod_{i=1}^{R}\phi_{q}(n_{i}, N_{i}, m_{i}).
\end{equation}
\end{proof}

\section{State transition Probability}
\label{sec:transprob}

The transition probability from an arbitrary state-action pair to another arbitrary state is written as
\begin	{equation}
\begin{split}
&\Pr \left( S[n+1]\ \big| \ S[n],A[n]\right)\\
&=\Pr\left(B_{k}[n+1]\;(\forall k \in P),B[n+1],\tau[n+1] \ \big| \ B_{k}[n]\;(\forall k\in P),B[n],\tau[n],A[n]\right)\\
&=\Pr\left(B_{k}[n+1]\; (\forall k \in P),B[n+1] \ \bigg| B_{k}[n]\;(\forall k \in P),B[n],\tau[n],A[n]\right)\\
&=\Pr\left(B[n+1] \ \big| \ \mathcal{I}[n]\right)\times\Pr\left(B_{k}[n+1]\; (\forall k \in P) \ \big| \ \mathcal{I}[n],B[n+1]\right)
\end{split}
\label{eq:CalProb}
\end{equation}
where $\mathcal{I}[n]=\{B_{k}[n]\; (\forall k \in P),B[n],\tau[n],A[n]\}$ is the complete local information set at $t_{n}$. For the sake of simplicity, we have not included the channel parameters $f_k$, $\epsilon_k$, $\eta_k$ in (\ref{eq:CalProb}). The first conditional probability in (\ref{eq:CalProb}) represents the belief of the receiver for the buffer state while the second is the belief of the receiver for the buffer states of the parents. Specifically, $\Pr(B[n+1] \ | \ \mathcal{I}[n])$ denotes the probability for the receiver to be in the buffer state $B[n+1]$ at the next decision time. Since the action $A[n]$ is given, the types of the requested packets and their transmission order is known. Using the remaining time information extracted from $\tau[n]$, along with the channel capacity $f_k, \forall k\in P$, the packet loss probabilities $(\epsilon_{k},\forall k\in P)$ and the delay distributions $(\eta_{k},\forall k\in P)$ in the incoming links, the probability that exactly $z_{ml}$ packets arrive in time (before their corresponding deadlines) for any packet type $(G_{m},l)\in A[n]$ can be calculated. Then, for a given $z_{ml}$ value, the problem reduces to a slightly more complex version of the coupon collector problem, as we have different classes of packets. The transition probabilities are calculated using the propositions presented in the Appendix \ref{app} which take into consideration the maximum number of innovative packets per packet type $(G_{m},l)$ and the size of the finite field $\mathbb{F}_q$. Regarding the second term in (\ref{eq:CalProb}), we assume that it is equal to one as the content of the servers is known to the receiver. This is due to the fact that the servers and the receivers know when the data from each generation becomes available to the servers.


\end{document}